\documentclass[12pt]{amsart}
\usepackage{epsfig}
\usepackage{amsfonts}
\usepackage{amsmath}
\usepackage{amssymb}
\newtheorem{thm}{Theorem}[section]

\newtheorem{cor}[thm]{Corollary}
\newtheorem{lem}[thm]{Lemma}
\newtheorem{prop}[thm]{Proposition}
\newtheorem*{prob*}{Problem}
\newtheorem*{thm*}{Theorem}

\theoremstyle{definition}
\newtheorem{defn}[thm]{Definition}

\newtheorem*{defn*}{Definition}
\newtheorem{rem}[thm]{Remark}

\newtheorem*{rem*}{Remark}
\numberwithin{equation}{section}

\newcommand{\C}{\mathbb C}

\newcommand{\R}{\mathbb R}
\newcommand{\Y}{\mathbb Y}

\newcommand{\eps}{\varepsilon}
\newcommand{\wt}{\widetilde}

\DeclareMathOperator{\Sp}{Sp}

\DeclareMathOperator{\Jacobi}{Jacobi}

\DeclareMathOperator{\const}{const}

\DeclareMathOperator{\E}{\mathbb{E}}

\DeclareMathOperator{\Product}{Product}

\DeclareMathOperator{\Markov}{Markov}
\DeclareMathOperator{\Macdonald}{Macdonald}

\begin{document}
\title[Product Matrix Processes via Symmetric Functions]
{\bf{Product Matrix Processes with Symplectic and Orthogonal Invariance via Symmetric Functions}}
\author{Andrew Ahn}
\address{Department of Mathematics, Massachusetts Institute of Technology,
77 Massachusetts Avenue
Cambridge, MA 02139-4307, USA}\email{ajahn@mit.edu}
\author{Eugene Strahov}
\address{Department of Mathematics, The Hebrew University of Jerusalem, Givat Ram, Jerusalem 91904, Israel}
\email{strahov@math.huji.ac.il}
\keywords{Products of  random matrices, singular value statistics, Macdonald symmetric functions, Gelfand pairs}
\commby{}
\begin{abstract}
We apply symmetric function theory to study random processes formed by singular values of products of truncations of Haar distributed symplectic and orthogonal matrices. These product matrix processes are degenerations of Macdonald processes introduced by Borodin and Corwin. Through this connection, we obtain explicit formulae for the distribution of singular values of a deterministic matrix multiplied by a truncated Haar orthogonal or symplectic matrix under conditions where the latter factor acts as a rank $1$ perturbation. Consequently, we generalize the recent Kieburg-Kuijlaars-Stivigny formula for the joint singular value density of a product of truncated unitary matrices to symplectic and orthogonal symmetry classes. Specializing to products of two symplectic matrices with a rank $1$ perturbative factor, we show that the squared singular values form a Pfaffian point process.
 \end{abstract}

\maketitle
\section{Introduction}
Products of random matrices enjoy a special seat at the intersection of numerous areas including dynamical systems \cite{Furstenberg,Pollicott}, neural networks \cite{Hanin}, and random matrices. In particular, the squared singular values are of interest, in connection with Lyapunov exponents in dynamical systems, exploding and vanishing gradients in neural networks, and as a particle system. For a variety of ensembles with unitary symmetry, such as products of Ginibre ensembles and truncated unitary matrices, the presence of exact formulas has enabled a fruitful exploration into these random processes, including asymptotic results under numerous limiting regimes, see Kuijlaars and Zhang \cite{KuijlaarsZhang},  Liu, Wang, and Zhang \cite{LiuWangZhang}, Akemann, Burda, and Kieburg \cite{AkemannBurdaKieburg}, Liu, Wang, and Wang \cite{LiuWangWang} and references therein.  The existence of exact formulas are due largely to the remarkably special nature of unitary symmetry. For analogous ensembles with orthogonal or symplectic symmetries instead, there is great scarcity of exact formulas even for just the product of two matrices.

One of the problems arising in the context of products of random matrices can  be formulated as follows. Assume that $X$ is a fixed matrix whose singular values are known, and
let $T$ be a random matrix of a compatible size. What can be said about the distribution of the singular values of the product $TX$? This problem is similar to
the randomized multiplicative Horn problems
discussed recently in the literature, see Forrester and Zhang \cite{ForresterZhang}, Zhang, Kieburg, and Forrester \cite{ZhangKieburgForrester}.
The randomized multiplicative Horn problems are versions of the classical Horn problem \cite{Horn} on finding the support of the eigenvalues of the sum $C=A+B$ of two fixed Hermitian matrices
$A$ and $B$ whose eigenvalues are given. One randomized multiplicative Horn problem is to study eigenvalues of the product matrix $C=A^{\frac{1}{2}}BA^{\frac{1}{2}}$ where $A$ and $B$
are given  by their fixed eigenvalues, and their diagonalizing Haar-distributed unitary matrices. For a progress on the multiplicative randomized Horn problems we refer
the reader to Refs. \cite{ForresterZhang, ZhangKieburgForrester}.

Concerning the distribution of the singular values of the product $TX$
the simplest  case  is that where $T$ has independent entries whose real and imaginary parts are independent  and have  a standard normal distribution, i.e. $T$ is a complex Ginibre matrix. In this situation the joint probability density function of the squared singular values can be explicitly computed, see Kuijlaars and Stivigny \cite{KuijlaarsStivigny}, Lemma 2.2,
and references therein to related  works.  Another known case is where $T$ is a submatrix (or truncation) of a Haar distributed unitary matrix. In this case  the squared singular values of $TX$  are distributed by a polynomial
ensemble, see Theorem 2.1 in Kieburg, Kuijlaars, and Stivigny \cite{KieburgKuijlaarsStivigny}.

The results mentioned above can be interpreted as those for the Markov transition kernel, from the squared singular values of a deterministic matrix $X$ to the squared singular values of $TX$, where $T$ is a random matrix. The explicit formulae for such Markov kernels can be used to study products of random matrices, and, in particular, \textit{product matrix processes}. Indeed, let us  assume that $X$ is a random matrix such that its squared singular values form a polynomial ensemble, and that $T$ is a complex Ginibre matrix  independent from $X$. Then the explicit formula for the  relevant Markov transition kernel implies  that the squared singular values of $TX$ form  a polynomial ensemble as well, see Theorem 2.1 in Kuijlaars \cite{Kuijlaars}. The subsequent application of  Theorem 2.1 in Kuijlaars \cite{Kuijlaars} gives the distribution of the squared singular values for the product
of an arbitrary number of independent Ginibre matrices (first obtained in the papers by Akemann, Kieburg, and Wei \cite{AkemannKieburgWei}, and by Akemann, Ipsen, and Kieburg \cite{AkemannIpsenKieburg}). Now, consider sequences $\left\{T_l\cdots T_1\right\}_{l=1}^m$ of products of independent Ginibre matrices instead of fixed products,
assume that each $T_l$ is of size $\left(n+\nu_l\right)\times\left(n+\nu_{l-1}\right)$ (where $l=1,\ldots, m; \nu_0=0,\nu_1\geq 0,\ldots, \nu_m\geq 0$), and for each $l=1,\ldots, m$ denote by $y_j^l$, $j=1,\ldots,n$, the squared singular values of $T_l\ldots T_1$. The configuration
$$
\left\{\left(l,y_j^l\right)\biggr|l=1,\ldots,m;\; j=1,\ldots,n\right\}
$$
of all these singular values generates a random point process on $\left\{1,\ldots,m\right\}\times \R_{>0}$ is called the \textit{Ginibre product process} in Strahov \cite{StrahovD}.
The application of the Markov transition kernel to this process gives the joint distribution of $y_j^l$ as a product of determinants.  The density formula implies that the Ginibre product matrix process is a discrete-time determinantal process, and accesses various asymptotic results, see Strahov \cite{StrahovD} for further details.

Matrix products with truncated unitary matrices can be studied in a similar way. An explicit formula for the Markov transition kernel, from the squared singular values
of a deterministic matrix $X$ to the squared singular values of $TX$ where $T$ is a truncated unitary matrix (see Kieburg, Kuijlaars, and Stivigny \cite{KieburgKuijlaarsStivigny}, Theorem 2.1) gives the joint densities of squared singular values for  products  of truncated unitary matrices.
These densities have   explicit determinantal forms in terms of Meijer $G$-functions, which leads to determinantal point processes formed by the squared singular values.
Through the remarkable fact that the product matrix process with truncated unitary matrices can be understood as a scaling limit of the Schur process, see  Borodin, Gorin, and Strahov \cite{BorodinGorinStrahov}, one can obtain determinantal formulas for (dynamical) correlation functions.

The existence of exact formulas for matrix products mentioned above is due to the special nature of \textit{unitary} symmetry.
In this article  we derive new formulas for the squared singular values
of truncated \textit{orthogonal} and \textit{symplectic} matrices, and for product matrix processes
formed by such truncations. By unitary, orthogonal and symplectic matrices we mean those taken from the classical unitary group $U(m)$, from the orthogonal group $O(m)$, and from the compact symplectic group $Sp(2m)$, respectively. The orthogonal matrices can be obtained by the restricting unitary matrices from $U(m)$ to real elements, and the symplectic matrices can be obtained by restricting unitary matrices from $U(m)$ to real quaternion elements. If we use a $2\times 2$ representation of the quaternions in terms of Pauli matrices, then rank $1$ pertubations for symplectic matrices are essentially rank 2 pertubations that respect the quaternion structure, and $k\times r $ truncations are essentially $2k\times 2r$ truncations.

Our main result is an explicit formula for the distribution of the squared singular values of a fixed matrix multiplied with a truncation of a Haar distributed orthogonal or symplectic matrix such that its number of rows is one less than that of the ambient Haar matrix, see Theorem \ref{THEOREMMarkovKernel} of the present paper. The condition on the number of rows effectively makes multiplication by this truncated Haar matrix behave as a rank $1$ pertubation at the level of singular values, in the sense that the singular values of the initial matrix and new matrix interlace. Although this assumption on row numbers may seem limiting, we note that the Markov transition kernel for an arbitrary truncated Haar orthogonal/symplectic matrix can be obtained by a composition of Markov transition kernels of rank $1$-pertubative type. Thus the rank $1$ perturbations we study can be viewed as the elementary building blocks for general truncated matrix transitions. A subsequent application of Theorem \ref{THEOREMMarkovKernel} gives the joint law for the product matrix processes with truncated symplectic and orthogonal matrices, see
Theorem \ref{THEOREMProductTruncatedProcessTHETA}. Theorem \ref{THEOREMProductTruncatedProcessTHETA} is used to derive two formulas for the joint probability density of the squared singular values of matrix products
constructed with truncated symplectic and orthogonal matrices, see Theorem \ref{THEOREMsvproduct} and Theorem \ref{THEOREMDENSITYJACKREPRESENTATION}. Theorem \ref{THEOREMsvproduct} presents the joint density
as an integral where the dimension of the integral depends only on the number of matrices in products, and Theorem \ref{THEOREMDENSITYJACKREPRESENTATION} gives the joint probability density of the squared singular values in terms of a Jack symmetric function with the appropriate parameter associated to the symmetry class.

We highlight a key structural result which we obtain in the symplectic case from our main result. We recast the two-product density for truncated symplectic matrices --- again, where one factor is a rank $1$-pertubative type --- in a determinantal form in terms of Meijer $G$-functions. The structure of this determinant indicates that it is a Pfaffian point process and we compute the correlation kernel via skew orthogonal polynomials. The methods we apply in the two product case are not stable under iteration, thus we are unable to see the Pfaffian structure for products of three or more matrices, if it exists. If there exists a Pfaffian structure for general products of truncated orthogonal and symplectic matrices, this may be due to special properties of the Jack functions in the parameters associated to these symmetry classes. Thus, one possible approach to find this structure in general may be exploit the Jack function representation of our density. We hope to explore this in a later work.

Our entry point to these results is a connection between Macdonald processes introduced in Borodin and Corwin \cite{BorodinCorwin}, and products of truncated orthogonal, unitary, and symplectic matrices. This connection is rooted in the fact that the zonal spherical functions associated to the Gelfand pairs $GL_n(\R)/O(n)$, $GL_n(\C)/U(n)$, $GL_n(\mathbb{H})/\Sp(2n)$ (where $\mathbb{H}$ denotes the skew field of real quaternions) are given by the Jack functions with appropriate parameters which are certain degenerations of the Macdonald symmetric functions. More concretely, the squared singular values of these products can be realized as a degeneration of certain Macdonald processes. In \cite{BorodinGorinStrahov}, this degeneration was applied in the unitary case, and \cite{Ahn} considered this degeneration for the other symmetry classes. Through this connection, we produce an explicit formula for the Markov transition kernel from squared singular value of a deterministic matrix $X$ to the squared singular values of a product $TX$ where $T$ is a truncated orthogonal or symplectic matrix, analogous to the result of \cite{KieburgKuijlaarsStivigny,Kuijlaars}.
With the Markov transition kernel, we compute the density in terms of Cauchy-type and Vandermonde determinants. Using generalized Dixon-type integrals, we are able to derive  an integral representation for the density where the dimension of the integral depends only on the number of products. By exploiting a variable-index symmetry in the Macdonald symmetric functions, we derive an alternative formula for the density in terms of the Jack symmetric functions.

Similar problems to those considered in this paper can be formulated for sums of random matrices, and for the Hermitised products $X_M\ldots X_2X_1AX_1^{T}X_2^{T}\ldots X_M^T$, where each $X_i$ is a real random matrix, and $A$ is real antisymmetric,
 see Kuijlaars and Rom$\acute{\mbox{a}}$n \cite{KuijlaarsRoman},
Kieburg \cite{Kieburg}, Kieburg,  Forrester, and  Ipsen \cite{KieburgForresterIpsen} and references therein for applications of the theory of spherical functions to these ensembles. An application of the theory of symmetric functions
to such sums and products is a possible topic for a future research.

While the focus of this article is on the structure and form of densities of the squared singular values of random matrix products, we mention a couple asymptotic applications. First, we note that the Dyson index $\beta$, where $\beta = 1,2,4$ corresponds to the orthogonal, unitary, and symplectic symmetry classes respectively, can be generalized to arbitrary $\beta > 0$ in our model, extending the idea of $\beta$-ensembles to the setting of matrix products, see e.g. \cite[Chapter 20]{AkemannBaikDiFrancesco}. This generalization is described in Section \ref{SECarbitrarybeta}. By taking $\beta \to \infty$, we find that the particles crystallize and the fluctuations of the particles are described by correlated Gaussians, see Section \ref{SECbeta=infty}. An analogous object was introduced by Gorin and Marcus \cite{GorinMarcus} for the corners process and studied in the $n\to\infty$ limit by Gorin and Kleptsyn \cite{GorinKleptsyn}. We expect that similar methods as in \cite{GorinKleptsyn} may be employed to study the $n\to\infty$ of the $\beta = \infty$ product process.

We mention one more application. Recall that our Theorem \ref{THEOREMDENSITYJACKREPRESENTATION} describes the density of squared singular values of matrix products in terms of the Jack function. With this density formula and integral representations for the Jack function \cite{Cuenca}, we can access the limiting behavior of the squared singular values in the regime where $n$ is fixed but the number of products tends to $\infty$. Here, we expect Lyapunov exponents to determine the asymptotic positions of the log singular values, and the fluctuations to be Gaussian. This regime was considered for products of $p$-adic random matrices in \cite{VanPeski} where the Lyaunov exponents and fluctuations were determined. A related regime for complex Ginibre and truncated unitary matrices was also considered in \cite{AkemannBurdaKieburg}, \cite{LiuWangWang}, \cite{Ahn}, where it was shown that if the matrix size $n$ and the number of products $p$ tend to infinity with $n/p \to 0$, then the largest log singular values tend to Lyapunov exponents with Gaussian fluctuations as well. We defer further details in this direction to a future article.

The remainder of this article is organized as follows. In Section \ref{SECmainresults}, we describe in more details the background behind our work and our main results. In Section \ref{SECkernel}, we prove our formula for the Markov transition and the first density formula. Section \ref{SECdixon} derives the second density formula via generalized Dixon integration. Section \ref{SECjack} derives the third density formula in terms of Jack functions. We derive the density in the case of products of two truncated orthogonal matrices and the corresponding Pfaffian correlation kernel in Section \ref{SECpfaffian}. In Section \ref{SECarbitrarybeta}, we give some details about the interpretation of our models for arbitrary $\beta > 0$. Lastly, in Section \ref{SECbeta=infty}, we take the $\beta \to \infty$ limit of our model. \\
\textbf{Acknowledgements.}
This work was supported by the BSF grant 2018248 ``Products of random matrices via the theory of symmetric functions''. A. Ahn was partially supported by National Science Foundation Grant DMS-1664619.

\section{Formulation of the problem and the main results} \label{SECmainresults}
The starting point of the present research is the well-known fact that the distribution of the squared singular values of a truncation
of a Haar distributed  matrix taken from unitary, symplectic, or orthogonal group is given by the Jacobi ensemble from Random Matrix Theory.
\subsection{The Jacobi ensembles related to truncated unitary, orthogonal, and symplectic matrices}
In this paper we are dealing with random matrix ensembles related to classical compact groups $U(m)$, $O(m)$, and $Sp(2m)$.
Here $U(m)$ denotes the group of $m\times m$ unitary matrices over $\C$, $O(m)$ denotes the group of $m\times m$ unitary matrices over $\R$, and $Sp(2m)$ denotes the group of $m\times m$ unitary matrices over the field
of real quaternions $\mathbb{H}$.  In what follows we refer to $U(m)$ as to the unitary group, to $O(m)$ as to the orthogonal group, and to $Sp(2m)$ as to the symplectic group.
Let $S$ be a Haar distributed matrix taken from the unitary group $U(m)$, or from the orthogonal group $O(m)$, or from the symplectic group $\Sp(2m)$.
Let the integers $k$, $r$ be chosen such that the condition $1\leq k,r\leq m$ is satisfied. The submatrix $T$ of $S$ defined by
$$
T=\left(
          \begin{array}{ccc}
            S_{1,1} & \ldots & S_{1,r} \\
            \vdots &  &  \\
            S_{k,1} & \ldots & S_{k,r} \\
          \end{array}
        \right)
$$
is called a $k\times r$ \textit{truncation} of $S$.
\begin{prop}\label{PROPOSITIONJacobiEnsembles}Let $\nu$ be a positive integer, and assume that $T$ is a $(n+\nu)\times n$-truncation of a Haar distributed matrix $S$ taken from the unitary group $U(m)$,
or from the orthogonal group $O(m)$, or from the symplectic group $\Sp(2m)$. In addition, assume that the condition
$$
m\geq 2n+\nu
$$
is satisfied. Then the distribution of  the eigenvalues $\left(x_1,\ldots,x_n\right)$ of $T^*T$ is given by the probability measure on $[0,1]^n$ defined by
\begin{equation}\label{DISTRIBUTIONONEMATRIX}
\begin{split}
&P_{n,\Jacobi}^{(\theta)}\left(x_1,\ldots,x_n\right)dx_1\ldots dx_n\\
&=\frac{1}{Z_{n,\Jacobi}^{(\theta)}}\prod\limits_{1\leq j<k\leq n}\left|x_j-x_k\right|^{2\theta}\prod\limits_{j=1}^n
\left(x_j\right)^{\theta\left(\nu+1\right)-1}\left(1-x_j\right)^{\theta(m-2n-\nu+1)-1}dx_1\ldots dx_n.
\end{split}
\end{equation}
Here $Z_{n,\Jacobi}^{(\theta)}$ is the normalization constant
given by the formula
\begin{equation}\label{ZJacobi}
Z_{n,\Jacobi}^{(\theta)}=n!\left(\Gamma(\theta)\right)^n\prod\limits_{j = \nu_1+1}^{m_1-n}\frac{\Gamma(\theta j)}{\Gamma(\theta(j+n))}
\prod\limits_{j=1}^n\frac{\Gamma(\theta\left(m-2n-\nu+j\right))\Gamma(\theta j)}{\left(\Gamma(\theta)\right)^2},
\end{equation}
$\theta=1$ in case
$T$ is a $(n+\nu)\times n$-truncation of a Haar distributed matrix  taken from the unitary group $U(m)$,
$\theta=\frac{1}{2}$ in case $T$ is a $(n+\nu)\times n$-truncation of a Haar distributed matrix
taken  from the orthogonal group $O(m)$, and $\theta=2$ in case $T$ is a $(n+\nu)\times n$-truncation of a Haar distributed matrix taken from the symplectic group $\Sp(2m)$.
\end{prop}
\begin{proof}
See Forrester \cite{Forrester}, Section 3.8.3.
\end{proof}
Proposition \ref{PROPOSITIONJacobiEnsembles} is a fundamental fact of Random Matrix Theory. In particular, Proposition \ref{PROPOSITIONJacobiEnsembles} implies that
the eigenvalues $\left(x_1,\ldots,x_n\right)$ of $T^*T$ form a determinantal point process in case  $T$ is a truncation of a unitary matrix, or Pfaffian point process in case $T$ is a truncation
of an orthogonal matrix, or a matrix from $Sp(2m)$. The correlation functions of such point processes can be found explicitly in terms of special functions, and different scaling limits can be obtained.

The question arises whether it is possible to extend the results of Proposition \ref{PROPOSITIONJacobiEnsembles} to \textit{products} of truncated matrices.
The paper by Kieburg, Kuijlaars, and Stivigny \cite{KieburgKuijlaarsStivigny} gives an answer to this question in the case of truncation of matrices taken from $U(m)$.
\subsection{The Kieburg-Kuijlaars-Stivigny theorem, and some of its  consequences}\label{SectionKKSTheoremAndItsConsequences}
In Ref. \cite{KieburgKuijlaarsStivigny} Kieburg, Kuijlaars, and Stivigny proved that the squared singular values of a fixed matrix multiplied
with a truncation of a Haar distributed unitary matrix form a polynomial ensemble (in the sense of Kuijlaars \cite{Kuijlaars}). Namely, the following Theorem holds true
\begin{thm}\label{THEOREMKKS} Assume that $S$ is a Haar distributed unitary matrix of size $m\times m$, and let $T$ be a $(n+\nu)\times l$ truncation of $S$.
In addition, let $X$ be a nonrandom matrix of size $l\times n$ such that the conditions
\begin{equation}\label{ConditionsForIndexes}
1\leq n\leq l\leq m,\;\; m\geq n+\nu+1
\end{equation}
are satisfied, and such that the eigenvalues $\left(x_1,\ldots,x_n\right)$ of $X^*X$ are pairwise distinct and nonzero. Then the vector
$\left(y_1,\ldots,y_n\right)$ of eigenvalues of $\left(TX\right)^*\left(TX\right)$ has density
$$
\const\left(\prod\limits_{j=1}^nx_j^{-m+n}\right)\left(\prod\limits_{j=1}^ny_j^{\nu}\right)\det\left[\left(x_k-y_j\right)_+^{m-n-\nu-1}\right]_{j,k=1}^n
\frac{\triangle(y)}{\triangle(x)},
$$
where $\triangle(x)=\prod_{1\leq i<j\leq n}\left(x_j-x_i\right)$ is the Vandermonde determinant, $\left(x-y\right)_+=\max\left(0,x-y\right)$, and the normalization constant only depends on $n$, $m$, and $\nu$, but is independent on $\left(x_1,\ldots,x_n\right)$.
\end{thm}
\begin{rem}The conditions $l\leq m$ and $m\geq n+\nu+1$ in (\ref{ConditionsForIndexes}) are restrictions on the size of matrix $S$, and on the size of matrix $T$. These conditions ensure that $T$ is a proper truncation of $S$.  The condition $l\geq n$ enables to avoid nonrandom singular values.
\end{rem}
Proposition \ref{PROPOSITIONJacobiEnsembles} (with $\theta=1$) and Theorem \ref{THEOREMKKS} enable the study of the \textit{product matrix process} associated with truncated unitary matrices
as a determinantal process, see Borodin, Gorin, and Strahov  \cite{BorodinGorinStrahov}. Let $U_1$, $\ldots$, $U_p$ be independent Haar distributed unitary matrices.
We assume that the size of each matrix $U_j$, $1\leq j\leq p$, is equal to $m_j\times m_j$, and denote by $T_j$ the truncation of $U_j$ of size $\left(n+\nu_j\right)
\times\left(n+\nu_{j-1}\right)$. Let us agree that the positive integers $n$, $\nu_1$, $\ldots$, $\nu_p$ are chosen in such a way that the conditions
\begin{equation}\label{Condition1}
m_1\geq 2n+\nu_1,
\end{equation}
and
\begin{equation}\label{Condition2}
m_j\geq n+\nu_j+1,\;\; 2\leq j\leq p
\end{equation}
are satisfied. Also, we agree that $\nu_0=0$. Denote by $\left(x_1^k,\ldots,x_n^k\right)$ the vector of the squared singular values of the product matrix $T_k\ldots T_1$.
Configurations $\left\{\left(k,x_j^k\right)\biggl\vert k=1,\ldots,p; j=1,\ldots,n\right\}$  form a point process on $\left\{1,\ldots,p\right\}\times R_{>0}$
called the product matrix process associated with the truncated unitary matrices. Proposition \ref{PROPOSITIONJacobiEnsembles} (with $\theta=1$) and Theorem \ref{THEOREMKKS}
lead to representation of the joint probability distribution of $\left(x_1^k,\ldots,x_n^k\right)$ (where $1\leq k\leq p$) in terms of the product of determinants. \begin{thm}\label{THEOREMProductTruncatedProcess}
Consider the product matrix process associated with truncated unitary matrices.
The joint probability distribution of $\left(x_1^k,\ldots,x_n^k\right)$ is given by
\begin{equation}\label{TruncatedProcessJD}
\begin{split}
&\frac{1}{Z_{n,p}}
\triangle\left(x^p\right)\\
&\times\prod\limits_{r=1}^{p-1}\det\left[\left(x_j^{r+1}\right)^{\nu_{r+1}}\left(x_k^{r}-x_j^{r+1}\right)_+^{m_{r+1}-n-\nu_{r+1}-1}\left(x_k^{r}\right)^{n-m_{r+1}}\right]_{k,j=1}^n\\
&\times\det\left[w_k^{(1)}\left(x_j^1\right)\right]_{k,j=1}^n
dx^1\ldots dx^n,
\end{split}
\end{equation}
where $(x-y)_+=\max\left(0,x-y\right)$, the Vandermonde determinant $\triangle\left(x^p\right)$ is defined by
$\triangle\left(x^p\right)=\prod\limits_{1\leq i<j\leq n}\left(x_j^p-x_i^p\right)$,
for $1\leq l\leq p$ we write $dx^l=dx^l_1\ldots dx^l_n$, and
\begin{equation}
w_k^{(1)}\left(x\right)=\left\{
                          \begin{array}{ll}
                            x^{\nu_1+k-1}(1-x)^{m_1-2n-\nu_1}, & \hbox{if}\;\; 0<x<1, \\
                            0, & \hbox{otherwise.}
                          \end{array}
                        \right.
\end{equation}
The normalization constant $Z_{n,p}$ can be written explicitly as
\begin{equation}
Z_{n,p}=\frac{\prod\limits_{j=1}^n\Gamma\left(m_1-2n-\nu_1+j\right)\Gamma(j)
\prod\limits_{k=2}^{p}\left(\Gamma\left(m_k-n-\nu_k\right)\right)^n}{\prod\limits_{k=1}^{p}\prod\limits_{j_k=1}^{m_k-n-\nu_k}
\left(j_k+\nu_k\right)_n}.
\end{equation}
Here $(a)_m=a(a+1)\ldots (a+m-1)$ stands for the Pochhammer symbol.
\end{thm}

It was shown in Ref. \cite{BorodinGorinStrahov} that the Eynard-Mehta theorem can be applied to  the joint probability distribution of $\left(x_1^k,\ldots,x_n^k\right)$,
and that the correlation functions have the structure of a determinantal point process on $\{1,\ldots,p\}\times [0,1]$. This process
can be understood as a scaling limit of the Schur process, the observation relating the product matrix processes with the theory
of $q$-distributed plane partitions, see Ref. \cite{BorodinGorinStrahov} for details.

It is instructive to consider separately the distribution of the squared singular values of the total product matrix $Y=T_p\ldots T_1$.
Proposition \ref{PROPOSITIONJacobiEnsembles}  and Theorem \ref{THEOREMKKS}
can be used to derive a formula for this distribution. Namely, the following result was first obtained in  Kieburg, Kuijlaars, and Stivigny \cite{KieburgKuijlaarsStivigny}.
\begin{thm}\label{THEOREMKKSPRODUCTMATRIXSINGULARVALEUES}Let $U_1$, $\ldots$, $U_p$ be independent Haar distributed unitary matrices. Assume that the size of each $U_j$, $2\leq j\leq p$, is equal to $m_j\times m_j$, and denote by $T_j$ the truncation of $U_j$ of size $\left(n+\nu_j\right)\times\left(n+\nu_{j-1}\right)$. If conditions (\ref{Condition1}), (\ref{Condition2}) are satisfied, then
the squared singular values $\left(y_1,\ldots,y_n\right)$ of $Y=T_p\ldots T_1$ have the joint probability density
\begin{equation}
\frac{1}{Z_{n,p}}\prod\limits_{1\leq j<k\leq n}\left(y_k-y_j\right)\det\left[w_k^{(p)}\left(y_j\right)\right]_{k,j=1}^n,
\end{equation}
where $Z_{n,p}$ is a normalization constant,
and
\begin{equation}
w_k^{(j)}(y)=\int\limits_0^1x^{\nu_j}(1-x)^{m_j-n-\nu_j-1}w_k^{(j-1)}\left(\frac{y}{x}\right)\frac{dx}{x}
\end{equation}
for $j=2,\ldots,p$.
\end{thm}
\begin{rem}
The function $w_k^{(l)}(x)$ can be expressed as a Meijer $G$-function,
\begin{equation}\label{TheWeightFunctions}
\begin{split}
&w_k^{(l)}(x)=c_lG_{l,l}^{l,0}\left(\begin{array}{cccc}
                                     m_l-n, & \ldots, & m_2-n, & m_1-2n+k \\
                                     \nu_l, & \ldots, & \nu_2, & \nu_1+k-1
                                   \end{array}
\biggl|x\right)\\
&=\frac{c_l}{2\pi i}\int\limits_C\frac{\Gamma\left(\nu_1+k-1+s\right)\prod_{j=2}^l\Gamma\left(\nu_j+s\right)}{\Gamma\left(m_1-2n+k+s\right)\prod_{j=2}^l
\Gamma\left(m_j-n+s\right)}x^{-s}ds,\;\;\; 0<x<1.
\end{split}
\end{equation}
In this formula $C$ denotes a positively oriented contour in the complex $s$-plane that starts and ends at $-\infty$ and encircles
the negative real axis. The constant $c_l$ in the formula for $w_k^{(l)}(x)$ can be written as
\begin{equation}\label{cl}
c_l=\Gamma\left(m_1-2n-\nu_1+1\right)\prod\limits_{j=2}^l\Gamma\left(m_j-n-\nu_j\right).
\end{equation}
\end{rem}
\subsection{Main results}
In this paper we obtain the extensions of the Kieburg-Kuijlaars-Stivigny theorem (Theorem \ref{THEOREMKKS}), and of  Theorems \ref{THEOREMProductTruncatedProcess} and \ref{THEOREMKKSPRODUCTMATRIXSINGULARVALEUES} to the products with truncations of Haar distributed orthogonal and symplectic matrices. As a consequence, we show that the squared singular values of certain products of two truncated symplectic matrices, where one factor acts as a ``rank $1$'' perturbation (described further below), is a Pfaffian point process and explicitly derive the correlation kernel.

We begin with the formula (which will be proved in Section \ref{SectionProofFirstMainTheorem})  for the distribution of the squared singular values of a fixed matrix multiplied with a truncation of a Haar distributed unitary, orthogonal, or symplectic matrix.
\begin{thm} \label{THEOREMMarkovKernel}
Assume that $S$ is a Haar distributed matrix taken from the unitary group $U(m)$, the orthogonal group $O(m)$, or from the symplectic group $Sp(2m)$.
Let $T$ be a $(n+\nu)\times l$ truncation of $S$, and assume that the conditions
\begin{equation}\label{ConditionsForIndexes1}
1\leq n\leq l\leq m,\;\; m= n+\nu+1
\end{equation}
are satisfied. Let $X$ be  a non-random matrix of size $l\times n$ with the squared singular values $\left(x_1,\ldots,x_n\right)$ such that
$$
0<x_1<\ldots<x_n<1.
$$
Then the joint distribution of the ordered eigenvalues $\left(y_1,\ldots,y_n\right)$ of $(TX)^*(TX)$ is given by the probability measure
\begin{equation}\label{TheJointDistributionTwoMatrices}
\frac{1}{\left(\Gamma(\theta)\right)^n}\frac{\Gamma\left(\theta\left(\nu+n+1\right)\right)}{\Gamma\left(\theta(\nu+1)\right)}
\left(\frac{\triangle(y)}{\triangle(x)}\right)^{\theta}\left(\det\left[\frac{1}{x_i-y_j}\right]_{i,j=1}^n\right)^{1-\theta}
\prod\limits_{i=1}^n\frac{y_i^{\theta(\nu+1)-1}}{x_i^{\theta(\nu+2)-1}}dy_i
\end{equation}
supported in
$$
0\le y_1\le x_1\le \ldots\le y_n\le x_n<1.
$$
Here $\theta=1$ corresponds to the multiplication with a truncation of a Haar distributed unitary matrix, $\theta=\frac{1}{2}$ corresponds to the multiplication with a truncation of a Haar distributed orthogonal matrix, and $\theta=2$ corresponds to the multiplication with a truncation of a Haar distributed symplectic matrix.
\end{thm}
\begin{rem}For the case of the orthogonal group  many ensembles  were studied in multivariate statistics,
see Muirhead \cite{Muirhead}. However, to the best of our knowledge  products of truncated  orthogonal matrices were not considered previously, and  our Theorem \ref{THEOREMMarkovKernel} is a new result.
\end{rem}
\begin{rem}
The condition $m = n + \nu + 1$ means that the number of rows of $T$ is one less than that of the ambient unitary matrix $U$. As a result, $T$ behaves as a ``rank-one'' perturbation which is reflected by the support satisfying the interlacing condition.
\end{rem}
\begin{rem}
(a) The joint distribution of the eigenvalues $\left(y_1,\ldots,y_n\right)$ can be also given by the probability measure
\begin{equation}\label{TheJointDistributionTwoMatrices1}
\frac{1}{\left(\Gamma(\theta)\right)^n}\frac{\Gamma\left(\theta\left(\nu+n+1\right)\right)}{\Gamma\left(\theta(\nu+1)\right)}\prod\limits_{i,j=1}^n|x_i-y_j|^{\theta-1}
\frac{\triangle(y)}{\left(\triangle(x)\right)^{2\theta-1}}
\det\left[\left(\left(x_k-y_j\right)_+\right)^{0}\right]_{j,k=1}^n
\prod\limits_{i=1}^n\frac{y_i^{\theta(\nu+1)-1}}{x_i^{\theta(\nu+2)-1}}dy_i
\end{equation}
supported in $[0,1]^n$. Indeed, for $0<y_1<\ldots<y_n<1$ and $0<x_1<\ldots<x_n<1$ it is not hard to check that
\begin{equation}
\det\left[\left(\left(x_k-y_j\right)_+\right)^{0}\right]_{j,k=1}^n=\left\{
                                                                     \begin{array}{ll}
                                                                       1, & \hbox{if}\;\;0<y_1<x_1<\ldots<y_n<x_n<1, \\
                                                                       0, & \hbox{otherwise.}
                                                                     \end{array}
                                                                   \right.
\end{equation}
Taking this into account, and applying the formula for the Cauchy determinant, we obtain (\ref{TheJointDistributionTwoMatrices}) from (\ref{TheJointDistributionTwoMatrices1}).\\
(b) We see that if $\theta=1$ (which corresponds to the multiplication with a truncation of a Haar distributed unitary matrix), then
the probability measure (\ref{TheJointDistributionTwoMatrices1}) turns into that given by the Kieburg-Kuijlaars-Stivigny theorem (see Theorem \ref{THEOREMKKS}) with
$m=n+\nu+1$.\\
(c) If $m>n+\nu+1$, then the methods of the present paper cannot be applied, and the extension of the Kieburg-Kuijlaars-Stivigny theorem
to truncations of orthogonal or symplectic matrices is an open question.
\end{rem}
In the same way as in the case of unitary matrices we introduce the product matrix processes associated with truncated orthogonal and truncated symplectic matrices.
Let $S_1$, $\ldots$, $S_p$ be independent Haar distributed symplectic or orthogonal matrices. As in the case of unitary matrices (see Section \ref{SectionKKSTheoremAndItsConsequences}) assume that the size of each matrix $S_j$, $1\leq j\leq p$, is equal to $m_j\times m_j$, and denote by
$T_j$ the $\left(n+\nu_j\right)\times\left(n+\nu_{j-1}\right)$ truncation of $S_j$. Let $\left(x_1^k,\ldots,x_n^k\right)$ be the vector of the squared singular values of
$T_k\ldots T_1$. We will refer to the point process on $\left\{1,\ldots,p\right\}\times\R_{>0}$ formed by configurations
$\left\{\left(k,x_j^k\right)\biggl\vert k=1,\ldots,p; j=1,\ldots,n\right\}$ as to the \textit{product matrix process with truncated symplectic/orthogonal matrices}.
\begin{thm}\label{THEOREMProductTruncatedProcessTHETA}Consider the product matrix process with truncated unitary matrices ($\theta=1$), or with truncated orthogonal matrices $\left(\theta=\frac{1}{2}\right)$, or with
truncated symplectic matrices $\left(\theta=2\right)$, and suppose that
the parameters $n$, $\nu_1$, $\ldots$, $\nu_p$ satisfy the conditions
\begin{equation}\label{c1}
m_1\geq 2n+\nu_1,
\end{equation}
and
\begin{equation}\label{c2}
m_k=n+\nu_k+1,\;\; 2\leq k\leq p.
\end{equation}
Let us agree that $\nu_0=0$.
Then the  joint probability distribution of $\left(x_1^k,\ldots,x_n^k\right)$ (where $1\leq k\leq p$) is given by the probability measure
\begin{equation}\label{ProductTruncatedProcessTHETA}
\begin{split}
&\frac{1}{Z_{n,p,\theta}}\left(\triangle\left(x^{p}\right)\right)^{\theta}
\prod\limits_{k=2}^p\biggl\{\left(\det\left(\frac{1}{x_i^{k-1}-x_j^{k}}\right)_{i,j=1}^n\right)^{1-\theta}\prod\limits_{i=1}^n
\frac{\left(x_i^{k}\right)^{\theta\left(\nu_k+1\right)-1}}{\left(x_i^{k-1}\right)^{\theta\left(\nu_k+2\right)-1}}\biggr\}\\
&\times\left(\triangle\left(x^{1}\right)\right)^{\theta}
\prod\limits_{i=1}^n\left(x_i^{1}\right)^{\theta\left(\nu_1+1\right)-1}\left(1-x_i^{1}\right)^{\theta\left(m_1-2n-\nu_1+1\right)-1}
\prod\limits_{k=1}^pdx^{k}
\end{split}
\end{equation}
supported in point configurations satisfying the condition
\begin{equation}
0\le x_1^k\le x_1^{k-1}\le\ldots\le x_n^{k}\le x_n^{k-1}\le 1
\end{equation}
for $2\leq k\leq p$. Here for $1\leq k\leq p$  we write $dx^k=dx^k_1\ldots dx^k_n$, and $Z_{n,p,\theta}$ denote the normalization
constant given by the formula
\begin{equation}\label{Znptheta}
Z_{n,p,\theta}=\left(\Gamma(\theta)\right)^{np}\prod\limits_{k=2}^p\frac{\Gamma\left(\theta\left(\nu_k+1\right)\right)}{\Gamma\left(\theta\left(\nu_k+n+1\right)\right)}
\prod\limits_{j=\nu_1+1}^{m_1-n}\frac{\Gamma(\theta j)}{\Gamma(\theta(j+n))}
\prod\limits_{j=1}^n\frac{\Gamma\left(\theta\left(m_1-2n-\nu_1+j\right)\right)\Gamma(\theta j)}{\Gamma(\theta)^2}.
\end{equation}
\end{thm}
Theorem \ref{THEOREMProductTruncatedProcessTHETA} will be proved in Section \ref{SectionSecondMainTheorem}.
Theorem \ref{THEOREMProductTruncatedProcessTHETA} can be used to derive the distribution of  squared singular values for a product matrix
formed by truncated orthogonal or symplectic matrices. Indeed, in order to obtain the distribution of $\left(x_1^p,\ldots,x_n^p\right)$
in Theorem \ref{THEOREMProductTruncatedProcessTHETA} it is enough to integrate the joint probability distribution (\ref{ProductTruncatedProcessTHETA}) over all configurations
$\left(x_1^p,\ldots,x_n^p\right)$ with $1\leq k\leq p-1$. Such an integration can be performed due to certain identities equivalent to those derived by Dixon \cite{Dixon}.
The result is the following Theorem which will be proved in Section \ref{SECdixon}.
\begin{thm}\label{THEOREMsvproduct}
Let $X=T_p\ldots T_1$, where each $T_j$ is the $\left(n+\nu_j\right)\times\left(n+\nu_{j-1}\right)$ truncation of a Haar distributed unitary, orthogonal, or symplectic matrix $S_j$ of size $m_j\times m_j$, $1\leq j\leq p$. Assume that conditions (\ref{c1}) and (\ref{c2}) are satisfied. In addition, assume that $\nu_{p+1}-\nu_1+1>0$.
Then the squared singular values
$\left(x_1,\ldots,x_n\right)$ of $X$ have the joint density
\begin{equation}\label{svproduct}
\frac{1}{Z_{n,p,\theta}W_{n,p}}\left(\triangle\left(x_1,\ldots,x_n\right)\right)^{2\theta}
I_{m_1;\nu_1,\ldots,\nu_p}^{\theta,n,p}\left(x_1,\ldots,x_n\right)\prod\limits_{i=1}^n\left(1-x_i\right)^{\theta\left(m_1-2n-\nu_1+p\right)-1}x_i^{\theta\left(\nu_1+1\right)-1},
\end{equation}
where  $\theta=1$ in case $T_1$,$\ldots$, $T_p$ are truncated unitary matrices, $\theta=\frac{1}{2}$ in  case $T_1$,$\ldots$, $T_p$ are truncated orthogonal matrices, or $\theta=2$ in case $T_1$,$\ldots$, $T_p$ are truncated symplectic matrices. The constant $Z_{n,p,\theta}$ is defined by equation (\ref{Znptheta}), and the constant
$W_{n,p}$ is given by
\begin{equation}
W_{n,p}=\prod\limits_{r=2}^p\frac{\Gamma\left(\theta\left(m_1-n-\nu_r\right)
\Gamma\left(\theta\left(\nu_r-\nu_1+1\right)\right)\right)}{\Gamma(\theta)^{n-r+2}\Gamma\left(\theta\left(m_1-2n-\nu_1+r-1\right)\right)}.
\end{equation}
The function $I_{m_1;\nu_1,\ldots,\nu_p}^{\theta,n,p}\left(x_1,\ldots,x_n\right)$ in equation (\ref{svproduct}) has the following integral representation
\begin{equation}
\begin{split}
&I_{m_1;\nu_1,\ldots,\nu_p}^{\theta,n,p}\left(x_1,\ldots,x_n\right)\\
&=\int\limits_{\Omega_1}dv_{1,1}\underset{\Omega_2}{\int\int}dv_{2,1}dv_{2,2}
\ldots\underset{\Omega_{p-1}}{\int\ldots\int}dv_{p-1,1}\ldots dv_{p-1,p-1}\triangle\left(v_{p-1,1},\ldots,v_{p-1,p-1}\right)\\
&\times\prod\limits_{r=1}^{p-2}\biggl\{
\left(\triangle\left(v_{r,1},\ldots,v_{r,r}\right)\right)^{2(1-\theta)}
\prod\limits_{i=1}^{r+1}\prod\limits_{j=1}^r
\left|v_{r+1,i}-v_{r,j}\right|^{\theta-1}
\prod\limits_{i=1}^r\left(v_{r,i}\right)^{\theta\left(\nu_r-\nu_{r+1}-1\right)}
\biggr\}\\
&\times\prod\limits_{i=1}^{p-1}\biggl\{\left(1-v_{p-1,i}\right)^{-\theta\left(m_1-2n-\nu_1+p-1\right)}
v_{p-1,i}^{\theta\left(\nu_p-\nu_1+1\right)-1}\prod\limits_{j=1}^n\left|x_j-v_{p-1,i}\right|^{-\theta}\biggr\},
\end{split}
\end{equation}
where the integration region $\Omega_1$ is defined by
\begin{equation}
\Omega_1=\left\{-\infty<v_{1,1}<0\right\},
\end{equation}
and the integration regions $\Omega_r$, $2\leq r\leq p-1$, are defined by
\begin{equation}
\Omega_r=\left\{-\infty<v_{r,1}<v_{r-1,1}<v_{r,2}<v_{r-1,2}<v_{r,3}<\ldots<v_{r,r-1}<v_{r-1,r-1}<v_{r,r}<0\right\}.
\end{equation}
\end{thm}
\begin{rem}
Assume that $p=2$ (the case corresponding to the product of two matrices). Then the function $I_{m_1;\nu_1,\ldots,\nu_p}^{\theta,n,p}\left(x_1,\ldots,x_n\right)$
takes the form
\begin{equation}
I_{m_1;\nu_1,\nu_2}^{\theta,n,p=2}\left(x_1,\ldots,x_n\right)
=\int\limits_{-\infty}^0dv(1-v)^{-\theta\left(m_1-2n-\nu_1+1\right)}
v^{\theta\left(\nu_2-\nu_1+1\right)-1}\prod\limits_{j=1}^n\left|x_j-v\right|^{-\theta},
\end{equation}
and the density of the distribution of the squared singular values $\left(x_1,\ldots,x_n\right)$ of $X=T_2T_1$ is proportional to
\begin{equation}
\begin{split}
&\left(\triangle\left(x_1,\ldots,x_n\right)\right)^{2\theta}
\prod\limits_{i=1}^n\left(1-x_i\right)^{\theta\left(m_1-2n-\nu_1+2\right)-1}x_i^{\theta\left(\nu_1+1\right)}\\
&\times\int\limits_{0}^{+\infty}dv(1+v)^{-\theta\left(m_1-2n-\nu_1+1\right)}
v^{\theta\left(\nu_2-\nu_1+1\right)-1}\prod\limits_{i=1}^n\left|x_i+v\right|^{-\theta}.
\end{split}
\end{equation}
\end{rem}
We have found an alternative representation for the distribution of the squared singular values of the total product matrix $X=T_p\ldots T_1$.
Namely, under certain
additional restrictions the joint singular value density of a product of truncated unitary, symplectic, or  orthogonal matrices can be expressed using the Jack symmetric functions.
This will be done in Section \ref{SECjack}. The result is the following
\begin{thm}\label{THEOREMDENSITYJACKREPRESENTATION}Let $X=T_p\ldots T_1$, where each $T_j$ is the $\left(n+\nu_j\right)\times\left(n+\nu_{j-1}\right)$ truncation of a Haar distributed unitary, orthogonal, or symplectic matrix $S_j$ of size $m_j\times m_j$, $1\leq j\leq p$. Assume that conditions (\ref{c1}) and (\ref{c2}) are satisfied.
Set
$$
M=\sum\limits_{i=1}^p\left(m_i-n-\nu_i\right),
$$
and assume there exists a partition $\mu$ with $l(\mu)\leq M$ such that the numbers
$$
\theta\left(\nu_1+1\right),\ldots,\theta\left(m_1-n\right),\ldots, \theta\left(\nu_p+1\right),\ldots,\theta\left(m_p-n\right)
$$
is a rearrangement of
$$
\mu_1+\theta(M-1)>\mu_2+\theta(M-2)>\ldots>\mu_M.
$$
Then the squared singular values $\left(x_1<\ldots<x_n\right)$ of $X$ have the joint density
\begin{equation} \label{JACKDensity}
\begin{split}
&\frac{1}{\hat{Z}_{n,p,\theta}}\frac{J_{\mu}\left(x_1,\ldots,x_n,1^{M-n};\theta\right)}{
J_{\mu}\left(1^M;\theta\right)}\\
&\times\prod\limits_{1\leq i<j\leq n}\left(x_j-x_i\right)^{2\theta}\prod\limits_{i=1}^n\left(1-x_i\right)^{\theta\left(M-n\right)
+\theta-1}\frac{dx_i}{x_i},
\end{split}
\end{equation}
where
\begin{itemize}
  \item $\theta=1$ in case $T_1$,$\ldots$, $T_p$ are truncated unitary matrices, $\theta=\frac{1}{2}$ in  case $T_1$,$\ldots$, $T_p$ are truncated orthogonal matrices, or $\theta=2$ in case $T_1$,$\ldots$, $T_p$ are truncated symplectic matrices;
  \item  $J_{\mu}\left(\ldots;\theta\right)$ stands for the Jack symmetric function with the Jack parameter $\theta$ parameterized by $\mu$;
  \item the constant $\hat{Z}_{n,p,\theta}$ is defined by
  \begin{equation}
  \hat{Z}_{n,p,\theta}=\prod\limits_{i=1}^p\prod\limits_{j=\nu_i+1}^{m_i-n}\frac{\Gamma(\theta j)}{\Gamma(\theta(j+n))}
  \prod\limits_{i=1}^n\frac{\Gamma\left(\theta\left(\sum\limits_{i=1}^p\left(m_i-n-\nu_i\right)-n+i\right)\Gamma(\theta i)\right)}{\Gamma(\theta)}.
  \end{equation}
\end{itemize}
\end{thm}
\begin{rem}Assume that $\theta=p=1$. In this case $M=m_1-n-\nu_1$, $\mu$ is a rectangular Young diagram with $m_1-n-\nu_1$ rows, and with $\nu_1+1$ boxes in each row.  Theorem \ref{THEOREMDENSITYJACKREPRESENTATION} gives the following expression for the distribution of the squared singular values of the truncated unitary matrix
\begin{equation}
\const s_{\mu}\left(x_1,\ldots,x_n;1^{m_1-2n-\nu_1}\right)\prod\limits_{1\leq i<j\leq n}\left(x_j-x_i\right)^2\prod\limits_{i=1}^n
\left(1-x_i\right)^{m_1-2n-\nu_1}\frac{dx_i}{x_i}.
\end{equation}
The combinatorial formula for the Schur functions can be applied, and we find
$$
s_{\mu}\left(x_1,\ldots,x_n;1^{m_1-2n-\nu_1}\right)=\prod\limits_{j=1}^nx_j^{\nu_1+1}.
$$
Thus we obtain formula (\ref{DISTRIBUTIONONEMATRIX}) (with $\theta=1$) from Theorem \ref{THEOREMDENSITYJACKREPRESENTATION}.
\end{rem}
Comparing the results of Theorem \ref{THEOREMDENSITYJACKREPRESENTATION} with those of Theorem \ref{THEOREMKKSPRODUCTMATRIXSINGULARVALEUES}  we obtain a representation of the Schur symmetric function (with a suitable choice of variables) in terms of Meijer G-functions. Namely, the following Corollary holds true
\begin{cor}Suppose $\mu$ is a partition such that $l(\mu)\leq p-1\leq M$ and $M-p+1\geq n$. Then
\begin{equation}\label{smain}
\frac{s_{\mu}\left(x_1,\ldots,x_n,1^{M-n}\right)}{s_{\mu}\left(1^M\right)}=\frac{1}{\prod_{i=1}^n\left(1-x_i\right)^{M-n}}
\frac{1}{\prod\limits_{1\leq i<j\leq n}\left(x_j-x_i\right)}\det\left[w_j\left(x_i\right)\right]_{i,j=1}^n,
\end{equation}
where
\begin{equation}
\begin{split}
&w_j(x)=\frac{\Gamma(M-n-p+2)\Gamma(M-n+j)}{\Gamma(M-n-p+1+j)}\\
&\times G_{p,p}^{p,0}\left(\begin{array}{cccc}
                                                                                  \mu_1+M, & \ldots, & \mu_{p-1}+M-p+2, & M-p-n+j+1 \\
                                                                                  \mu_1+M-1, &\ldots, & \mu_{p-1}+M-p+1, & j-1
                                                                                \end{array}
\biggl| x\right).
\end{split}
\end{equation}
\end{cor}
\begin{rem}Consider the case where $p=N-k+1$. In this case $l(\lambda)\leq N-k$, and  the functions $w_j(x)$ can be written as contour integrals,
\begin{equation}
w_j(x)=\frac{(-1)^{N-k}\Gamma(N-k+j)}{2\pi i\Gamma(j)}
\oint\limits_{\Sigma}\frac{x^sds}{\left[s-\left(\lambda_1+N-1\right)\right]\ldots\left[s-\left(\lambda_{N-k}+N-(N-k)\right)\right][s-j+1]},
\nonumber
\end{equation}
where $1\leq j\leq k$, and the contour $\Sigma$ encloses all the singularities of the integrand. Taking this into account we see
that equation (\ref{smain}) gives
\begin{equation}
\begin{split}
&\frac{s_{\lambda}\left(x_1,\ldots,x_k,1^{N-k}\right)}{s_{\lambda}\left(1^N\right)}=\frac{(-1)^{N-k}}{\prod_{i=1}^k\left(1-x_i\right)^{N-k}
\prod\limits_{1\leq i<j\leq k}\left(x_j-x_i\right)}\frac{\prod_{j=1}^k\Gamma(N-k+j)}{\prod_{j=1}^k\Gamma(j)}\\
&\times\frac{1}{(2\pi i)^k}\oint\limits_{\Sigma_1}\ldots\oint\limits_{\Sigma_k}
\left|\begin{array}{cccc}
        \frac{1}{s_1} & \frac{1}{s_1-1} & \ldots & \frac{1}{s_1-k+1} \\
        \frac{1}{s_2} & \frac{1}{s_2-1} & \ldots & \frac{1}{s_2-k+1} \\
        \vdots &  &  &  \\
        \frac{1}{s_k} & \frac{1}{s_k-1} & \ldots & \frac{1}{s_k-k+1}
      \end{array}
\right|
\frac{x_1^{s_1}\ldots x_k^{s_k}ds_1\ldots ds_k}{\prod\limits_{j=1}^k\left[s_j-\left(\lambda_1+N-1\right)\right]\ldots\left[s_j
-\left(\lambda_{N-k}+N-(N-k)\right)\right]}.
\end{split}
\nonumber
\end{equation}
The determinant in the integrand can be computed explicitly using the formula for the Cauchy determinant,
\begin{equation}
\left|\begin{array}{cccc}
        \frac{1}{s_1} & \frac{1}{s_1-1} & \ldots & \frac{1}{s_1-k+1} \\
        \frac{1}{s_2} & \frac{1}{s_2-1} & \ldots & \frac{1}{s_2-k+1} \\
        \vdots &  &  &  \\
        \frac{1}{s_k} & \frac{1}{s_k-1} & \ldots & \frac{1}{s_k-k+1}
      \end{array}
\right|=\frac{\prod\limits_{1\leq i<j\leq k}\left(s_j-s_i\right)\prod\limits_{1\leq i<j\leq k}(j-i)}{\prod\limits_{i,j=1}^k\left(s_i-j+1\right)}.
\end{equation}
Since $\prod\limits_{1\leq i<j\leq k}(j-i)=\prod\limits_{j=1}^k\Gamma(j)$, and $\prod_{j=1}^k\Gamma(N-k+j)=\prod_{j=1}^k(N-j)!$,  we obtain
\begin{equation}
\begin{split}
&\frac{s_{\lambda}\left(x_1,\ldots,x_k,1^{N-k}\right)}{s_{\lambda}\left(1^N\right)}=\frac{\prod_{j=1}^k(N-j)!}{\prod_{i=1}^k\left(x_i-1\right)^{N-k}
\prod\limits_{1\leq i<j\leq k}\left(x_j-x_i\right)}\\
&\times\frac{1}{(2\pi i)^k}\oint\limits_{\Sigma_1}\ldots\oint\limits_{\Sigma_k}
\frac{\prod\limits_{1\leq i<j\leq k}\left(s_j-s_i\right)x_1^{s_1}\ldots x_k^{s_k}ds_1\ldots ds_k}{\prod\limits_{j=1}^k\left[s_j-\left(\lambda_1+N-1\right)\right]
\ldots\left[s_j-\left(\lambda_{N}+N-N\right)\right]},
\end{split}
\end{equation}
where $\lambda_{N-k+1}=\ldots=\lambda_N=0$. This formula is equivalent to the result of Gorin and Panova \cite{GorinPanova}, as it can be seen from equations (1.5) and (1.6)
in Gorin and Panova \cite{GorinPanova}.
\end{rem}
\subsection{The correlation functions for singular values of the product of two truncated symplectic matrices}\label{SectionCorrelationFunctionsStatements}
Let $X=T_2T_1$, where $T_1$ is the $\left(n+\nu_1\right)\times n$ truncation of a Haar distributed symplectic matrix $S_1$ of size $m_1\times m_1$, and $T_2$ is the $\left(n+\nu_2\right)\times\left(n+\nu_1\right)$ truncation of a Haar distributed symplectic matrix $S_2$ of size $m_2\times m_2$.
In this particular case we find that the density of squared singular values $\left(x_1,\ldots,x_n\right)$ of $X$ can be written as a determinant.
\begin{prop}\label{PROPOSITIONDENSITYDETERMINANT}
The density $P_{n,\Product}^{(2)}\left(x_1,\ldots,x_n\right)$ of squared singular values $\left(x_1,\ldots,x_n\right)$ of $X=T_2T_1$ is equal to
\begin{equation}\label{DensityProduct2}
P_{n,\Product}^{(2)}\left(x_1,\ldots,x_n\right)=\frac{1}{Z_{n,\Product}^{(2)}}
\det\left[\begin{array}{ccc}
            1 & \ldots & x_1^{2n-1} \\
            \vdots &  & \vdots \\
            1 & \ldots & x_n^{2n-1} \\
            W_1^{(2)}\left(x_1\right) & \ldots  & W_{2n}^{(2)}\left(x_1\right) \\
            \vdots &  & \vdots \\
            W_1^{(2)}\left(x_n\right) & \ldots & W_{2n}^{(2)}\left(x_n\right)
          \end{array}
\right],
\end{equation}
where $Z_{n,\Product}^{(2)}$ is the normalization constant,
\begin{equation}
Z_{n,\Product}^{(2)}=\frac{\Gamma\left(2\left(\nu_2+1\right)\right)}{\Gamma\left(2\left(\nu_2+n+1\right)\right)}
\prod\limits_{j=\nu_1+1}^{m_1-n}\frac{\Gamma(2j)}{\Gamma(2(j+n))}\prod\limits_{j=1}^m\Gamma\left(2\left(m_1-2n-\nu_1+j\right)\right)\Gamma(2j),
\end{equation}
 and
\begin{equation}
W_j^{(2)}(x)=G_{2,2}^{2,0}\left(\begin{array}{cc}
                                                                                                 2\nu_2+2 & 2\left(m_1-2n+1\right)+j-1 \\
                                                                                                 2\nu_2+1 & 2\nu_1+j-1
                                                                                               \end{array}
\biggr|x\right).
\end{equation}
This density is over $0<x_1<\ldots<x_n<1$.
\end{prop}
The proof of Proposition \ref{PROPOSITIONDENSITYDETERMINANT} will be given in Section \ref{FirstPropositionOnCorrelationFunctions}.
Assume that we have a probability measure on $\R^n$ which can be written as
\begin{equation}\label{SymplecticTypeEnsemble}
P_n\left(x_1,\ldots,x_n\right)dx_1\ldots dx_n=\frac{1}{Z_n}\det\left(\phi_j\left(x_k\right)\;\psi_j\left(x_k\right)\right)_{1\leq k\leq n,\;\;0\leq j\leq 2n-1}dx_1\ldots dx_n,
\end{equation}
where $\phi_j$, $\psi_j$ are certain functions, and $Z_n$ is the normalizing constant. We will refer to (\ref{SymplecticTypeEnsemble}) as to
a \textit{symplectic-type ensemble}.  Define the correlation kernel, $\mathbb{K}_n(x,y)$, of the symplectic-type ensemble (\ref{SymplecticTypeEnsemble})
as a $2\times 2$ matrix valued kernel of the operator $\mathbb{K}_n$ for which the following condition is satisfied
\begin{equation}\label{CorrelationKernelSymplecticTypeEnsemble}
\left(\int\ldots\int P_n\left(x_1,\ldots,x_n\right)\prod\limits_{i=1}^n\left(1+f\left(x_i\right)\right)dx_1\ldots dx_n\right)^2=\det\left(I+\mathbb{K}_nf\right).
\end{equation}
Here $\mathbb{K}_n$ denotes the operator on $L^2$ and $f$ is the operator of multiplication by that function.
Equation (\ref{DensityProduct2}) implies that the squared singular values of the product of two truncated symplectic matrices form a symplectic type ensemble.
The standard methods of Random Matrix Theory enable us to obtain in Section \ref{SectionSecondPropositionOnCorrelationFunctions} the following
\begin{prop}\label{PROPOSITIONK}
The correlation kernel for the density $P_{n,\Product}^{(2)}\left(x_1,\ldots,x_n\right)$ defined by equation (\ref{DensityProduct2})
can be written as
\begin{equation}
\mathbb{K}_{n,\Product}(x,y)=\left(\begin{array}{cc}
                 K_{n,\Product}^{(1,1)}(x,y) &  K_{n,\Product}^{(1,2)}(x,y) \\
                 K_{n,\Product}^{(2,1)}(x,y) &  K_{n,\Product}^{(2,2)}(x,y)
               \end{array}
\right),
\end{equation}
where
\begin{equation}\label{K11}
K_{n,\Product}^{(1,1)}(x,y)=\sum\limits_{k,l=0}^{2n-1}G_{2,2}^{2,0}\left(\begin{array}{cc}
                                                                                                 2\nu_2+2 & 2\left(m_1-2n+1\right)+k \\
                                                                                                 2\nu_2+1 & 2\nu_1+k
                                                                                               \end{array}
\biggr|x\right)q_{k,l}^{\Product}y^l,
\end{equation}
\begin{equation}\label{K12}
\begin{split}
&K_{n,\Product}^{(1,2)}(x,y)\\
&=-\sum\limits_{k,l=0}^{2n-1}
G_{2,2}^{2,0}\left(\begin{array}{cc}
                                                                                                 2\nu_2+2 & 2\left(m_1-2n+1\right)+k \\
                                                                                                 2\nu_2+1 & 2\nu_1+k
                                                                                               \end{array}
\biggr|x\right)q_{k,l}^{\Product}G_{2,2}^{2,0}\left(\begin{array}{cc}
                                                                                                 2\nu_2+2 & 2\left(m_1-2n+1\right)+l \\
                                                                                                 2\nu_2+1 & 2\nu_1+l
                                                                                               \end{array}
\biggr|y\right),
\end{split}
\end{equation}
\begin{equation}\label{K21}
\begin{split}
&K_{n,\Product}^{(2,1)}(x,y)=\sum\limits_{k,l=0}^{2n-1}x^kq_{k,l}^{\Product}y^l,
\end{split}
\end{equation}
and
\begin{equation}\label{K22}
K_{n,\Product}^{(2,2)}(x,y)=-\sum\limits_{k,l=0}^{2n-1}x^kq_{k,l}^{\Product}G_{2,2}^{2,0}\left(\begin{array}{cc}
                                                                                                 2\nu_2+2 & 2\left(m_1-2n+1\right)+l \\
                                                                                                 2\nu_2+1 & 2\nu_1+l
                                                                                               \end{array}
\biggr|y\right).
\end{equation}
Here $Q^{\Product}=\left(q_{i,j}^{\Product}\right)_{i,j=0}^{2n-1}$ is the inverse of $C^{\Product}=\left(c_{i,j}^{\Product}\right)_{i,j=0}^{2n-1}$ defined by
\begin{equation}\label{cproduct}
c_{i,j}^{\Product}=\frac{\left(j-i\right)}{\left(2\nu_2+i+2\right)\left(2\nu_2+j+2\right)}\frac{\Gamma\left(i+j+2\nu_1+1\right)}{\Gamma\left(
2\left(m_1-2n+1\right)+i+j+1\right)}.
\end{equation}
\end{prop}
We see that in order to obtain explicit formulae for the matrix entries of the kernel $\mathbb{K}_{n,\Product}(x,y)$ we need to find the inverse of the matrix
$C^{\Product}=\left(c_{i,j}^{\Product}\right)_{i,j=0}^{2n-1}$ defined by equation (\ref{cproduct}). This can be done using the following result that will be proved in Section \ref{SectionProofInverse}.
\begin{prop}\label{PROPOSITIONInverse}
The inverse of the matrix
$$
C=\left((j-i)\frac{\Gamma(a+i+j)}{\Gamma(a+b+i+j+1)}\right)_{i,j=0}^{2n-1}
$$
is the matrix $Q=\left(Q_{i,j}\right)_{i,j=0}^{2n-1}$, where
\begin{equation}\label{InverseJacobiType}
\begin{split}
Q_{i,j}=\frac{(-1)^{i+j}\Gamma(b+1)}{\Gamma(a+i)\Gamma(a+j)}&\sum\limits_{k=0}^{n-1}\sum\limits_{l=0}^k\frac{2^{4k}}{2^{4l}}
\frac{(a+b+4l-1)(a+b+4k+1)\Gamma(a+2l)\Gamma(a+1+2k)}{\Gamma(l+1)\Gamma\left(\frac{a}{2}+\frac{b}{2}+l\right)
\Gamma\left(\frac{a+1}{2}+l\right)\Gamma\left(\frac{b+1}{2}+l\right)}\\
&\times\Theta\left(k+1\right)\Theta\left(\frac{a+1}{2}+k\right)\Theta\left(\frac{b+1}{2}+k\right)
\Theta\left(\frac{a}{2}+\frac{b}{2}+k\right)\\
&\times\biggl(\Gamma(a+b+2l+i-1)\Gamma(a+b+2k+j)\left(\begin{array}{c}
                                                                                          2l \\
                                                                                          i
                                                                                        \end{array}
\right)
\left(\begin{array}{c}
                                                                                          2k+1 \\
                                                                                          j
                                                                                        \end{array}
\right)\\
&-\Gamma(a+b+2l+j-1)\Gamma(a+b+2k+i)\left(\begin{array}{c}
                                                                                          2l \\
                                                                                          j
                                                                                        \end{array}
\biggr)
\left(\begin{array}{c}
                                                                                          2k+1 \\
                                                                                          i
                                                                                        \end{array}
\right)\right),
\end{split}
\end{equation}
and where $\Theta(x)=\Gamma(x)/\Gamma(2x)$.
\end{prop}
Proposition (\ref{PROPOSITIONK}) together with Proposition (\ref{PROPOSITIONInverse}) enable us to give
explicit formulae for the matrix entries of the kernel $\mathbb{K}_{n,\Product}(x,y)$.
\begin{thm}\label{THEOREMCORRELATIONKERNELFINAL}Set $a_1=2\nu_1+1$, $a_2=2\nu_2+1$, $b_1=2\left(m_1-2n-\nu_1\right)+1$, and define
\begin{equation}
P_m^{\Product}(x)=\sum\limits_{i=0}^m(-1)^i\left(\begin{array}{c}
                                        m \\
                                        i
                                      \end{array}
\right)\frac{(a_2+i+1)\Gamma(a_1+b_1+m+i-1)}{\Gamma(a_1+i)}x^i,
\end{equation}
and
\begin{equation}
Q_p^{\Product}(y)=\sum\limits_{j=0}^p(-1)^j\left(\begin{array}{c}
                                        p \\
                                        j
                                      \end{array}
\right)\frac{\left(a_2+j+1\right)\Gamma(a_1+b_1+p+j-1)}{\Gamma(a_1+j)}
G_{2,2}^{2,0}\left(\begin{array}{cc}
                                                                                a_2+1 & b_1+j+a_1 \\
                                                                                 a_2  &    a_1+j-1
                                                                              \end{array}
\biggl|y\right).
\end{equation}
With these notation the matrix entries of the kernel $\mathbb{K}_{n,\Product}(x,y)$ can be written as
\begin{equation}\label{KProductSum11}
\begin{split}
K_{n,\Product}^{(1,1)}(x,y)=&\Gamma(b_1+1)\sum\limits_{k=0}^{n-1}\sum\limits_{l=0}^k\frac{2^{4k}}{2^{4l}}
\frac{(a_1+b_1+4l-1)(a_1+b_1+4k+1)\Gamma(a_1+2l)\Gamma(a_1+1+2k)}{\Gamma(l+1)\Gamma\left(\frac{a_1}{2}+\frac{b_1}{2}+l\right)
\Gamma\left(\frac{a_1+1}{2}+l\right)\Gamma\left(\frac{b_1+1}{2}+l\right)}\\
&\times\Theta\left(k+1\right)\Theta\left(\frac{a_1+1}{2}+k\right)\Theta\left(\frac{b_1+1}{2}+k\right)
\Theta\left(\frac{a_1}{2}+\frac{b_1}{2}+k\right)\\
&\times\left(Q_{2l}^{\Product}(x)P_{2k+1}^{\Product}(y)-Q_{2k+1}^{\Product}P_{2l}(y)^{\Product}\right),
\end{split}
\end{equation}
\begin{equation}\label{KProductSum12}
\begin{split}
K_{n,\Product}^{(1,2)}(x,y)=&-\Gamma(b_1+1)\sum\limits_{k=0}^{n-1}\sum\limits_{l=0}^k\frac{2^{4k}}{2^{4l}}
\frac{(a_1+b_1+4l-1)(a_1+b_1+4k+1)\Gamma(a_1+2l)\Gamma(a_1+1+2k)}{\Gamma(l+1)\Gamma\left(\frac{a_1}{2}+\frac{b_1}{2}+l\right)
\Gamma\left(\frac{a_1+1}{2}+l\right)\Gamma\left(\frac{b_1+1}{2}+l\right)}\\
&\times\Theta\left(k+1\right)\Theta\left(\frac{a_1+1}{2}+k\right)\Theta\left(\frac{b_1+1}{2}+k\right)
\Theta\left(\frac{a_1}{2}+\frac{b_1}{2}+k\right)\\
&\times\left(Q_{2l}^{\Product}(x)Q_{2k+1}^{\Product}(y)-Q_{2k+1}^{\Product}Q_{2l}(y)^{\Product}\right),
\end{split}
\end{equation}
\begin{equation}\label{KProductSum21}
\begin{split}
K_{n,\Product}^{(2,1)}(x,y)=&\Gamma(b_1+1)\sum\limits_{k=0}^{n-1}\sum\limits_{l=0}^k\frac{2^{4k}}{2^{4l}}
\frac{(a_1+b_1+4l-1)(a_1+b_1+4k+1)\Gamma(a_1+2l)\Gamma(a_1+1+2k)}{\Gamma(l+1)\Gamma\left(\frac{a_1}{2}+\frac{b_1}{2}+l\right)
\Gamma\left(\frac{a_1+1}{2}+l\right)\Gamma\left(\frac{b_1+1}{2}+l\right)}\\
&\times\Theta\left(k+1\right)\Theta\left(\frac{a_1+1}{2}+k\right)\Theta\left(\frac{b_1+1}{2}+k\right)
\Theta\left(\frac{a_1}{2}+\frac{b_1}{2}+k\right)\\
&\times\left(P_{2l}^{\Product}(x)P_{2k+1}^{\Product}(y)-P_{2k+1}^{\Product}P_{2l}(y)^{\Product}\right),
\end{split}
\end{equation}
and
\begin{equation}\label{KProductSum22}
\begin{split}
K_{n,\Product}^{(2,2)}(x,y)=&-\Gamma(b_1+1)\sum\limits_{k=0}^{n-1}\sum\limits_{l=0}^k\frac{2^{4k}}{2^{4l}}
\frac{(a_1+b_1+4l-1)(a_1+b_1+4k+1)\Gamma(a_1+2l)\Gamma(a_1+1+2k)}{\Gamma(l+1)\Gamma\left(\frac{a_1}{2}+\frac{b_1}{2}+l\right)
\Gamma\left(\frac{a_1+1}{2}+l\right)\Gamma\left(\frac{b_1+1}{2}+l\right)}\\
&\times\Theta\left(k+1\right)\Theta\left(\frac{a_1+1}{2}+k\right)\Theta\left(\frac{b_1+1}{2}+k\right)
\Theta\left(\frac{a_1}{2}+\frac{b_1}{2}+k\right)\\
&\times\left(P_{2l}^{\Product}(x)Q_{2k+1}^{\Product}(y)-P_{2k+1}^{\Product}Q_{2l}(y)^{\Product}\right).
\end{split}
\end{equation}
\end{thm}
The proof of this Theorem will be given in Section \ref{SectionFinalTheoremProof}.
\section{Proofs of Theorem \ref{THEOREMMarkovKernel} and Theorem \ref{THEOREMProductTruncatedProcessTHETA}} \label{SECkernel}
\subsection{The Markov kernel associated with the Macdonald measure on Young diagrams}
In this section we use the notation of Macdonald \cite{Macdonald}.  Let $\Lambda$ be the algebra of symmetric functions over the field of complex numbers $\C$.
Let $P_{\lambda}(x;q,t)$, $Q_{\lambda}(x;q,t)$ where $q,t\in(0,1)$ denote the ordinary and dual $(q,t)$-Macdonald symmetric functions respectively indexed by Young diagrams $\lambda$.
We will generally suppress the $q$ and $t$ and write $P_{\lambda}(x)$, $Q_{\lambda}(x)$ or $P_{\lambda}$, $Q_{\lambda}$ unless the presence of the $q $ and $t$ is pertinent.

The Macdonald symmetric functions form a basis for $\Lambda$. Suppose we have two sets of variables $A=\left(a_1,a_2,\ldots\right)$ and $B=\left(b_1,b_2,\ldots\right)$. From the Cauchy identity for Macdonald symmetric functions we obtain that
\begin{equation}
\sum\limits_{\lambda\in\Y}M_{\Macdonald}\left(\lambda;A,B\right)=1,
\end{equation}
where $M_{\Macdonald}\left(\lambda;A,B\right)$ is defined by
\begin{equation}
M_{\Macdonald}\left(\lambda;A,B\right)=\frac{1}{\Pi\left(A;B\right)}P_{\lambda}\left(A\right)Q_{\lambda}\left(B\right) = \frac{1}{\Pi\left(A;B\right)}Q_{\lambda}\left(A\right)P_{\lambda}\left(B\right),
\end{equation}
and
\begin{equation}
\Pi\left(A;B\right)=\prod\limits_{i,j}\frac{\left(ta_ib_j;q\right)_{\infty}}{\left(a_ib_j;q\right)_{\infty}},\;\; \left(u;q\right)_{\infty}=\prod\limits_{i=1}^{\infty}\left(1-q^{i-1}u\right).
\end{equation}
If both sets of variables $A=\left(a_1,a_2,\ldots\right)$ and $B=\left(b_1,b_2,\ldots\right)$ give positive specializations of the algebra $\Lambda$ of symmetric functions then
$M_{\Macdonald}\left(\lambda;A,B\right)$ can be understood as a probability measure on the set of all Young diagrams $\Y$. We emphasize that $M_{\Macdonald}\left(\lambda;A,B\right)$
depends on the Macdonald parameters $q$ and $t$. In what follows we will refer to $M_{\Macdonald}\left(\lambda;A,B\right)$ as to \textit{the Macdonald measure} on Young diagrams.

The skew Macdonald symmetric functions $P_{\lambda/\mu}$, $Q_{\lambda/\mu}$ are defined by
\begin{equation}\label{SkewMacdonald}
P_{\lambda}\left(A,B\right)=\sum\limits_{\mu\in\Y}P_{\lambda/\mu}\left(A\right)P_{\mu}\left(B\right),\;\;
Q_{\lambda}\left(A,B\right)=\sum\limits_{\mu\in\Y}Q_{\lambda/\mu}\left(A\right)Q_{\mu}\left(B\right),
\end{equation}
as  identities on $\Lambda \otimes \Lambda$. Equations (\ref{SkewMacdonald}) suggest to introduce the \textit{Markov kernel}
$K_{\Markov}\left(\lambda,\mu;A,B\right)$ for the Macdonald measure  $M_{\Macdonald}\left(\lambda;A,B\right)$ by the formula
\begin{equation}
K_{\Markov}\left(\lambda,\mu;A,B\right)=\frac{1}{\Pi\left(A;B\right)}\frac{P_{\lambda}\left(B\right)}{P_{\mu}\left(B\right)}Q_{\lambda/\mu}\left(A\right) =\frac{1}{\Pi\left(A;B\right)}\frac{Q_{\lambda}\left(B\right)}{Q_{\mu}\left(B\right)}P_{\lambda/\mu}\left(A\right),
\end{equation}
where in the last equality we have exploited the relationship
$$
Q_{\lambda/\mu}=\frac{\left\langle P_{\mu},P_{\mu}\right\rangle}{\left\langle P_{\lambda},P_{\lambda}\right\rangle}P_{\lambda/\mu}.
$$
Here $\left\langle .,.\right\rangle$ is the scalar product in the algebra of symmetric functions defined as in Macdonald \cite{Macdonald}, Chapter VI, equation (1.5).
Using equations (\ref{SkewMacdonald}) we obtain
\begin{equation}
\sum\limits_{\mu\in\Y}K_{\Markov}\left(\lambda,\mu;C,B\right)M_{\Macdonald}\left(\mu;A,B\right)=M_{\Macdonald}\left(\lambda;C\sqcup A,B\right),
\end{equation}
where $C=\left(c_1,c_2,\ldots\right)$ is an additional sequence of variables, and $C\sqcup A$ denotes the union of two collections $C$ and $A$ of independent variables.
In addition,
\begin{equation}
\sum\limits_{\lambda\in\Y}K_{\Markov}\left(\lambda,\mu;A,B\right)=1,
\end{equation}
as it follows from the summation formula
\begin{equation}\label{SkewMacdonaldProperty}
\sum\limits_{\lambda\in\Y}P_{\lambda/\mu}\left(B\right)Q_{\lambda/\nu}\left(A\right)=\Pi\left(A;B\right)\sum\limits_{\tau\in\Y}Q_{\mu/\tau}\left(A\right)P_{\nu/\tau}\left(B\right),
\end{equation}
see Macdonald \cite{Macdonald}, VI. 7. If $A$ and $B$ give positive specializations of the algebra $\Lambda$ of symmetric functions then
$K_{\Markov}\left(\lambda,\mu;A,B\right)$ can be understood as a probability measure on the set of all Young diagrams $\Y$ parameterized by $A$, $B$, and $\mu$.

We note a remarkable identity which relates the variables and indices of Macdonald symmetric functions which will be used in the sequel. For any partitions $\lambda,\nu \in \Y$ of length $\le n$, we have
\begin{align} \label{DUALITY}
\frac{P_\lambda(q^{\nu_1} t^{n-1},q^{\nu_2} t^{n-2},\ldots,q^{\nu_n})}{P_\lambda(q^{n-1},q^{n-2},\ldots,1)} = \frac{P_\nu(q^{\lambda_1} t^{n-1},q^{\lambda_2} t^{n-2},\ldots,q^{\lambda_n})}{P_\nu(q^{n-1},q^{n-2},\ldots,1)},
\end{align}
see \cite[Chapter VI, (6.6)]{Macdonald}.

\subsection{The convergence of the Markov kernel to the distribution of the squared singular values of $TX$}
Assume that $S$ is a Haar distributed matrix taken from the unitary group $U(m)$ (in this case we say that the Jack parameter $\theta$ is equal to 1), from the orthogonal group $\left(\theta=\frac{1}{2}\right)$, or from the symplectic group $\left(\theta=2\right)$. Let $T$ be a $(n+\nu)\times l$ truncation of $S$, and assume that the conditions
$m\geq 2n+\nu$ and $m\geq n+\nu+1$ are satisfied. Let $X$ be a non-random matrix of size $l\times n$ with squared singular values $\left(x_1,\ldots,x_n\right)$ such that
$$
0<x_1<\ldots<x_n<1.
$$
In addition, assume that the sets of variables $A$ and $B$ are given by
\begin{equation}\label{AB}
A=\left(t^{\nu+1},t^{\nu+2},\ldots,t^{m-n}\right),\;\; B=\left(1,t,\ldots,t^{n-1}\right).
\end{equation}
If $A$ is given by equation (\ref{AB}), then $K_{\Markov}\left(\lambda,\mu;A,B\right)$ is concentrated on the Young diagrams with $n$ rows or less, so we can assume that $\lambda$ has $n$ rows at most. The next Proposition says that the distribution of the eigenvalues of $(TX)^*(TX)$ can be obtained by a limiting procedure from the Markov kernel
$K_{\Markov}\left(\lambda,\mu;A,B\right)$.
\begin{prop} \label{PROPOSITIONMarkovKernel}
Let $\lambda$ be a random Young diagram whose distribution is defined by $K_{\Markov}\left(\lambda,\mu;A,B\right)$. Assume that the Macdonald parameters $q$ and $t$ depend on $\epsilon>0$, and are given by
\begin{equation}
q=q(\epsilon)=e^{-\epsilon},\;\; t=t(\epsilon)=\left(q(\epsilon)\right)^{\theta}=e^{-\theta\epsilon}.
\end{equation}
Let $\mu(\epsilon)=\left(\mu_1(\epsilon),\ldots,\mu_n(\epsilon)\right)$ be a Young diagram with $n$ rows, and suppose that the length of each row of $\mu$
depends on $\epsilon$ in such a way that the limits
\begin{equation}
\underset{\epsilon\rightarrow 0+}{\lim}\left(\left(q(\epsilon)\right)^{\mu_1(\epsilon)}\right)=x_1,\ldots,\underset{\epsilon\rightarrow 0+}{\lim}\left(\left(q(\epsilon)\right)^{\mu_n(\epsilon)}\right)=x_n
\end{equation}
exist. In addition, let $y(\epsilon)=\left(y_1(\epsilon),\ldots,y_n(\epsilon)\right)$ be a random configuration associated with the Young diagram $\lambda$ as
$$
y_1(\epsilon)=e^{-\epsilon\lambda_1},\ldots,y_n(\epsilon)=e^{-\epsilon\lambda_n}.
$$
As $\epsilon\rightarrow 0+$, the distribution of $y_1(\epsilon)$, $\ldots$, $y_n(\epsilon)$ will coincide with that of squared singular values  of $TX$.
\end{prop}

Proposition \ref{PROPOSITIONMarkovKernel} follows from Proposition 3.9 in the previous work of the first author \cite{Ahn}. Since this is our crucial link with symmetric function theory, we provide an alternative proof here. We recall the Jack symmetric functions which are the workhorse behind this proposition. In the limit $q,t \to 1$ such that $t = q^\theta$, we define
\[ J_\lambda(x_1,\ldots,x_n;\theta) = \lim_{q\to 1} P_\lambda(x_1,\ldots,x_n) \]
for $\lambda \in \Y$, see \cite[Chapter VI, Section 10]{Macdonald}. The Jack functions for $\ell(\lambda) \le n$ form a basis for the space of symmetric polynomials in $n$-variables.

There are three key ingredients. The first is a connection between the Jack functions and products of matrices.

\begin{lem} \label{LEMzonal}
Suppose $V$ is a Haar distributed $O(n), U(n), \Sp(2n)$ for $\theta = 1/2,1,2$ respectively. Let $W$ and $X$ be deterministic $n\times n$ matrices.
Then
\begin{equation}\label{AverageZonal}
\E J_{\kappa}\left(X^*V^*W^*WVX;\theta\right)
=\frac{J_{\kappa}\left(X^*X;\theta\right)J_{\kappa}\left(W^*W;\theta\right)}{J_{\kappa}\left(1^n;\theta\right)},
\end{equation}
where $\kappa\in\Y$, $\ell(\kappa) \le n$, and $J_{\kappa}\left(A;\theta\right)$ denotes the value of the Jack polynomial on the eigenvalues of $A$.
\end{lem}
\begin{proof}
Equation (\ref{AverageZonal}) follows immediately from the functional relation for zonal spherical functions, and from the fact that the zonal spherical functions for the relevant Gelfand pairs can be written in terms of symmetric functions.
Indeed, consider first the case corresponding to $\theta=1$. Let $G=GL\left(n;\C\right)$, and let $K$ be the unitary group $U(n)$. It is known that $\left(G,K\right)$ is a Gelfand pair,
see \cite[Chapter VII (\S 5)]{Macdonald}. Denote by $\Omega_{\lambda}(x)$, $x\in G$, the zonal spherical function associated with $(G,K)$. Then we have
\begin{equation}\label{SF}
\Omega_{\lambda}(x)=\frac{s_{\lambda}\left(x^*x\right)}{s_{\lambda}\left(1^n\right)},\;\; x\in G,
\end{equation}
where by $s_{\lambda}\left(x^*x\right)$ we mean the Schur polynomial $s_{\lambda}$ evaluated on eigenvalues of $x^*x$. Since $\Omega_{\lambda}(x)$ is the zonal spherical function it satisfies the functional relation
\begin{equation}
\int\limits_{U(n)}\Omega_{\lambda}(xky)dk=\Omega_{\lambda}(x)\Omega_{\lambda}(y),
\end{equation}
for all $x,y\in GL(n,\C)$. This relation can be rewritten more explicitly as
\begin{equation}\label{us}
\int\limits_{U(n)}s_{\lambda}\left((xky)^*(xky)\right)dk=\frac{s_{\lambda}\left(x^*x\right)s_{\lambda}\left(y^*y\right)}{s_{\lambda}\left(1^n\right)},
\end{equation}
which is equivalent to the statement of the Lemma  for $\theta=1$. If $G=GL\left(n,\R\right)$, $K=O(n)$ (the case corresponding to $\theta=\frac{1}{2}$),
or $G=GL(n,\mathbb{H})$, $K=Sp(2n)$ (the case corresponding to $\theta=2$), then the zonal spherical functions for the Gelfand pairs $(GL(n,\R),O(n))$, $(GL(n,\mathbb{H}),Sp(2n))$ are given by the same equation (\ref{SF}) with the Schur polynomials replaced by the Jack polynomials with $\theta=1/2$ or $\theta=2$
respectively, see \cite[Chapter VII (3.24)]{Macdonald} and \cite[Chapter VII (6.20)]{Macdonald}. As a result, one gets $\theta=1/2$ and $\theta=2$
analogues of equation (\ref{us}) which can be interpreted as in the statement of Lemma \ref{LEMzonal}.
\end{proof}
The second key ingredient is a set of integral formulas known as the Selberg integral and its generalizations. These identities can be found in \cite[(12.3),(12.46),(12.143)]{Forrester}.

\begin{lem} \label{LEMselberg}
For any $a,b,\theta > 0$, we have
\begin{align*}
S_n(a,b,\theta) & := \int_0^1 \cdots \int_0^1 \prod_{1 \le j \le k \le n} |u_j - u_k|^{2\theta} \prod_{j=1}^n (u_j)^a (1 - u_j)^b dx_j \\
& = \prod_{j=0}^{n-1} \frac{\Gamma(a + 1 + j\lambda) \Gamma(b + 1 + j\lambda) \Gamma(1 + (j+1)\theta)}{\Gamma(a + b + 2 + (N + j - 1)\theta) \Gamma(1 + \theta)}.
\end{align*}
Furthermore if $\kappa \in \Y$, then
\begin{align*}
& \frac{1}{S_n(a,b,\theta)} \int_0^1 du_1 \cdots du_N \prod_{j=1}^n (u_j)^a (1 - u_j)^b J_\kappa(u_1,\ldots,u_n;\theta) \prod_{1 \le j < k \le n} |u_j - u_k|^{2\theta} \\
& \quad \quad = J_\kappa((1)^n;\theta) \prod_{j=1}^n \frac{\Gamma(a + \theta(n - j) + 1 + \kappa_j)}{\Gamma(a + b + \theta(2n - j - 1) + 2 + \kappa_j)} \frac{\Gamma(a + b + \theta(2n - j - 1) + 2)}{\Gamma(a + \theta(n - j) + 1)}.
\end{align*}
\end{lem}

The final ingredient relates the squared singular values of products of square matrices with that of rectangular matrices, assuming distributional invariance with respect to $O(m),U(m)$, or $\Sp(2m)$.

\begin{lem} \label{LEMswitch_product}
Let $n,n_1,n_2,m \in \mathbb{Z}_+$ such that $n \le n_i \le m$, $i = 1,2$. Suppose $T$ and $\wt{T}$ are respectively $n_2 \times n_1$ and $n_2\times n$ truncations of a random Haar $O(m), U(m)$, or $\Sp(2m)$ matrix if $\theta = 1/2,1,2$ respectively. Let $X$ be a fixed $n_1 \times n$ matrix and $\wt{X} = (X^*X)^{1/2}$. If $\sigma(A) \in \R^N$ denotes the singular values of a matrix $A$, then $\sigma(TX) \overset{d}{=} \sigma(\wt{T} \wt{X})$.
\end{lem}
\begin{rem}The result of Lemma \ref{LEMswitch_product} was proved in a more general setting in Ipsen and Kieburg \cite{IpsenKieburg}. We give an independent proof below for the reader's convenience.
\end{rem}
\begin{proof}
Let $P_{a\times b}$ denote the $a\times b$ matrix with the $\min(a,b)\times \min(a,b)$ identity matrix in the upper left corner and $0$ elsewhere. The singular value decomposition of $X$ gives
\[ X = U P_{n_1\times n} \Sigma V^* \]
where $\Sigma = \mathrm{diag}(\sigma(X))$, $U,V$ are Haar distributed orthogonal, unitary, or symplectic matrices depending on whether $\theta = 1/2,1,$ or $2$ with the appropriate dimensions. Then
\[ (TX)^*(TX) = V \Sigma P_{n\times n_1} U^* T^* T U P_{n_1\times n} \Sigma V^*. \]
Observe that
\[ P_{n\times n_1} U^* T^* T U P_{n_1\times n} \overset{d}{=} \wt{T}^* \wt{T} \overset{d}{=} V^* \wt{T}^* \wt{T} V \]
using the fact that the distributions of $T$ and $\wt{T}$ are invariant under right translation by orthogonal, unitary, symplectic matrices for $\theta = 1/2,1,2$ respectively. Thus
\[ (TX)^*(TX) \overset{d}{=} V \Sigma V^* \wt{T}^* \wt{T} V \Sigma V^* = (\wt{T} \wt{X})^* \wt{T} \wt{X} \]
which proves the lemma.
\end{proof}

We are now ready to prove Proposition \ref{PROPOSITIONMarkovKernel}.

\begin{proof}[Proof of Proposition \ref{PROPOSITIONMarkovKernel}.]
We first recast the statement of Proposition \ref{PROPOSITIONMarkovKernel} in terms of square matrices. By Lemma \ref{LEMswitch_product}, the distribution for the squared singular values of $TX$ is the same as that of $\wt{T} \wt{X}$ where
\[ \wt{X} = (X^*X)^{1/2} \]
and $\wt{T}$ is a $(n+\nu)\times n$ truncation of $S$. The squared singular values of $\wt{X}$ are $(x_1,\ldots,x_n)$ and the squared singular values $(u_1,\ldots,u_n)$ of $\wt{T}$ are distributed as the Jacobi ensemble \eqref{DISTRIBUTIONONEMATRIX}. By right invariance,
$$
\wt{T} \overset{d}{=} \wt{T} V,
$$
where $V$ is a Haar distributed orthogonal, unitary, or symplectic matrix depending on $\theta=1/2$, $\theta=1$, or $\theta=2$ respectively. By Lemma \ref{LEMzonal}, the squared singular values $(y_1,\ldots,y_n)$ of $\wt{T} \wt{X}$ satisfy
\[ \E J_\kappa(y_1,\ldots,y_n;\theta) = \frac{J_\kappa(x_1,\ldots,x_n;\theta)}{J_\kappa(1^n;\theta)} \E J_\kappa(u_1,\ldots,u_n;\theta), \]
where the expectation in the right-hand side is over the Jacobi ensemble.
By Lemma \ref{LEMselberg}, the latter expectation is
\begin{align*}
& \int_0^1 \! \cdots \! \int_0^1 \prod_{1 \le j < k \le n} |u_j - u_k|^{2\theta} \frac{J_\kappa(u_1,\ldots,u_n;\theta) \prod_{j=1}^n (u_j)^{\theta(\nu + 1) - 1} (1 - u_j)^{\theta(m - 2n - \nu + 1) - 1} du_i}{S_n(\theta(\nu+1)-1,\theta(m - 2n - \nu + 1) - 1,\theta)}
\end{align*}
and can be evaluated so that
\begin{align} \label{MOMENTS}
\E J_\kappa(y_1,\ldots,y_n;\theta) = J_\kappa(x_1,\ldots,x_n;\theta) \prod_{j=1}^n \frac{\Gamma(\theta(\nu + n - j + 1) + \kappa_j)}{\Gamma(\theta(m - j + 1) + \kappa_j)} \frac{\Gamma(\theta(m - j + 1))}{\Gamma(\theta(\nu + n - j + 1))}.
\end{align}
We may view the distribution of $(y_1,\ldots,y_n)$ as being on the set of ordered real numbers $\{y_1 \ge \cdots \ge y_n\}$. Since the distribution of $(y_1,\ldots,y_n)$ is compactly supported, the Jack functions for $\ell(\kappa) \le n$ form a basis for symmetric polynomials in $n$-variables, and symmetric polynomials separate points in the set $\{y_1 \ge \cdots \ge y_n\}$, the Stone-Weierstrass theorem implies that the expectations \eqref{MOMENTS} determine the distribution of $(y_1,\ldots,y_n)$. Thus, to complete the proof of our proposition, it suffices to show that
\[ \lim_{\epsilon \to 0} \E J_\kappa(y_1(\epsilon),\ldots,y_n(\epsilon);\theta) = \E J_\kappa(y_1,\ldots,y_n;\theta), \]
where $y_1(\epsilon),\ldots,y_n(\epsilon)$ in the left-hand side is the random configuration introduced in the statement of Proposition \ref{PROPOSITIONMarkovKernel}.
Using the fact that
\[ \lim_{\epsilon \to 0} P_\kappa(z_1 t^{n-1}(\epsilon),z_2 t^{n-2}(\epsilon),\ldots,z_n;q(\epsilon),t(\epsilon)) = J_\kappa(z_1,\ldots,z_n;\theta) \]
uniformly over compact sets, we see that it is enough to show
\begin{align} \label{MACDONALDCONVERGENCE}
\lim_{\epsilon \to 0} \E P_\kappa(y_1(\epsilon)t^{n-1},\ldots,y_n(\epsilon);q,t) = \E J_\kappa(y_1,\ldots,y_n;\theta).
\end{align}
where $\E J_\kappa(y_1,\ldots,y_n;\theta)$ is given explicitly by equation (\ref{MOMENTS}). Let $C_{\kappa} = (q^{\kappa_1} t^{n-1},\ldots,q^{\kappa_n})$ for any Young diagram $\kappa$ with $n$ rows, and compute
\begin{align*}
\E P_\kappa(y_1(\epsilon)t^{n-1},\ldots,y_n(\epsilon)) &= \frac{1}{\Pi(A;B)} \sum_{\lambda \in \Y} P_\kappa(C_{\lambda}) \frac{P_\lambda(B)}{P_\mu(B)} Q_{\lambda/\mu}(A) \\
&= \frac{1}{\Pi(A;B)} \frac{P_\kappa(B)}{P_\mu(B)} \sum_{\lambda \in \Y} P_\lambda(C_{\kappa}) Q_{\lambda/\mu}(A) \\
&= \frac{\Pi(A;C_{\kappa})}{\Pi(A;B)} \frac{P_\kappa(B)}{P_\mu(B)} P_\mu(C_{\kappa}) = P_\kappa(C_{\mu}) \frac{\Pi(A;C_{\kappa})}{\Pi(A;B)}
\end{align*}
where $A$, $B$ are defined by equation (\ref{AB}), the second equality uses \eqref{DUALITY}, the third uses the specialization of \eqref{SkewMacdonaldProperty} to the case $\nu=\emptyset$, and the final uses \eqref{DUALITY} again. We now observe that the right hand side converges to the right hand side of \eqref{MOMENTS}. Thus we have shown \eqref{MACDONALDCONVERGENCE}, completing our proof.
\end{proof}

\subsection{Proof of Theorem \ref{THEOREMMarkovKernel}}\label{SectionProofFirstMainTheorem}
If $m-n=\nu+1$ (as we have assumed in the statement of Theorem \ref{THEOREMMarkovKernel}), then $A$ turns into the list containing $t^{\nu+1}$ only.
By Proposition \ref{PROPOSITIONMarkovKernel}, we establish the Markov kernel formula by showing that
\begin{align} \label{eq:prekernel}
\frac{(t^{\nu+1};q)_\infty}{(t^{\nu+n+1};q)_\infty} \cdot \frac{Q_\lambda(1,t,\ldots,t^{n-1})}{Q_\mu(1,t,\ldots,t^{n-1})} P_{\lambda/\mu}(t^{\nu+1})
\end{align}
converges to \eqref{TheJointDistributionTwoMatrices} as $\eps \to 0$, where $\lambda_i = -\eps^{-1} \log y_i$, $\mu_i = -\eps^{-1} \log y_i'$, $q = e^{-\eps}$, and $t = q^\theta$. A similar computation can be found in \cite{BorodinGorin} for the $\beta$-Jacobi corners process.

We use the following facts repeatedly: as $q \to 1$ we have
\begin{align} \label{eq:pochhammer}
(i)~\frac{(q^au;q)_\infty}{(q^bu;q)_\infty} \to (1 - u)^{b-a}, \quad (ii)~ \frac{(q^a;q)_\infty}{(q^b;q)_\infty} \to \frac{\Gamma(b)}{\Gamma(a)} \eps^{b-a}
\end{align}
where the former holds uniformly over compact subsets of $0 < u < 1$, see Lemma 2.4 in Borodin and Gorin \cite{BorodinGorin}.

For $\ell(\lambda) = n$, and $M\geq n$ we have
\begin{align} \label{MACDONALDFormula}
\begin{split}
P_\lambda(1,\ldots,t^{M-1}) &= t^{\sum_{i=1}^n (i-1)\lambda_i} \prod_{1 \le i < j \le M} \frac{(q^{\lambda_i - \lambda_j} t^{j-i};q)_\infty}{(q^{\lambda_i - \lambda_j} t^{j-i+1};q)_\infty} \frac{(t^{j-i+1};q)_\infty}{(t^{j-i};q)_\infty} \\
&= t^{\sum_{i=1}^n (i-1)\lambda_i} \prod_{1 \le i < j \le n} \frac{(q^{\lambda_i - \lambda_j} t^{j-i};q)_\infty}{(q^{\lambda_i-\lambda_j} t^{j-i+1};q)_\infty} \prod_{i=1}^n \prod_{j=n+1}^M \frac{(q^{\lambda_i} t^{j-i};q)_\infty}{(q^{\lambda_i} t^{j-i+1};q)_\infty} \prod_{\substack{i < j \\ 1 \le i \le n \\ 1 \le j \le M}} \frac{(t^{j-i+1};q)_\infty}{(t^{j-i};q)_\infty},
\end{split}
\end{align}
see  the proof of Theorem 2.8 in Borodin and Gorin \cite{BorodinGorin},  and Macdonald \cite{Macdonald}, Chapter VI, equation (6.11).
If $q = e^{-\eps}$, $t = q^\theta$ and $\lambda_i = - \eps^{-1} \log y_i$, then as $\eps \to 0$ we have
\begin{align} \label{ASYMPTOTICSt}
\begin{gathered}
t^{\sum_{i=1}^n \lambda_i(i-1)} \to \prod_{i=1}^n y_i^{\theta(i-1)}, \quad \frac{(q^{\lambda_i - \lambda_j} t^{j-i};q)_\infty}{(q^{\lambda_i - \lambda_j} t^{j-i+1};q)_\infty} \to (1 - y_i/y_j)^\theta, \\
\frac{(q^{\lambda_i} t^{j-i};q)_\infty}{(q^{\lambda_i} t^{j-i+1};q)_\infty} \to (1 - y_i)^\theta, \quad \frac{(t^{j-i+1};q)_\infty}{(t^{j-i};q)_\infty} \sim \frac{\Gamma(\theta(j-i))}{\Gamma(\theta(j-i+1))} \eps^{-\theta}.
\end{gathered}
\end{align}
From \cite{Macdonald}, Chapter VI, equation (6.19) we have
\[ \frac{Q_\lambda}{P_\lambda} = b_\lambda = \prod_{1 \le i \le j \le \ell(\lambda)} \frac{f(q^{\lambda_i - \lambda_j} t^{j-i})}{f(q^{\lambda_i - \lambda_{j+1}} t^{j-i})}, \quad f(u) = \frac{(tu;q)_\infty}{(qu;q)_\infty}. \]
We obtain
\[ b_\lambda \sim \prod_{i=1}^n \frac{f(1)}{(1 - y_i/y_{i+1})^{1-\theta}} \frac{(1 - y_i/y_{i+1})^{1-\theta}}{(1 - y_i/y_{i+2})^{1-\theta}} \cdots \frac{(1 - y_i/y_n)^{1-\theta}}{(1 - y_i)^{1-\theta}} = \prod_{i=1}^n \frac{f(1)}{(1 - y_i)^{1-\theta}} \sim \frac{\eps^{n(1-\theta)}}{\Gamma(\theta)^n} \prod_{i=1}^n (1 - y_i)^{\theta-1}. \]
These asymptotics imply
\begin{align} \label{eq:Q}
Q_\lambda(1,\ldots,t^{n-1}) \sim \frac{\eps^{-\theta n(n-1)/2 + n(1-\theta)}}{\Gamma(\theta)^n} \prod_{1 \le i < j \le n} \frac{\Gamma(\theta(j-i))}{\Gamma(\theta(j-i+1))} (y_j - y_i)^\theta \prod_{i=1}^n (1 - y_i)^{\theta-1}.
\end{align}
Given $\lambda, \mu \in \Y$ such that $\ell(\lambda),\ell(\mu) \le n$, let $\mu \prec \lambda$ denote the interlacing relation
\begin{equation}\label{interlacingXXX}
\lambda_1 \ge \mu_1 \ge \cdots \ge \lambda_n \ge \mu_n.
\end{equation}
In order to find the asymptotics of $P_{\lambda/\mu}\left(t^{\nu+1}\right)$ we use the combinatorial formula for the skew Macdonald symmetric functions
$P_{\lambda/\mu}$ (see \cite[Chapter VI, (6.19)]{Macdonald}) representing these functions as sums over all column-strict (skew) tableaux of shape $\lambda-\mu$. Restricting $P_{\lambda/\mu}$ to the Macdonald polynomial in a single variable $x$, we obtain
\[ P_{\lambda/\mu}(x;q,t) = \psi_{\lambda/\mu}(x) \delta_{\mu \prec \lambda} \]
where $\delta_{\mu \prec \lambda}$ is equal to $1$ in case condition (\ref{interlacingXXX}) is satisfied, and is equal to zero otherwise.  Here
\[ \psi_{\lambda/\mu}(x) = x^{|\lambda| - |\mu|} \prod_{1 \le i \le j \le \ell(\mu)} \frac{f(q^{\mu_i - \mu_j} t^{j-i}) f(q^{\lambda_i - \lambda_{j+1}} t^{j-i})}{f(q^{\lambda_i - \mu_j} t^{j-i}) f(q^{\mu_i - \lambda_{j+1}} t^{j-i})}, \]
where $|\lambda|$ denotes the number of boxes in the Young diagram $\lambda$.
Taking $\ell(\lambda) = \ell(\mu) = n$, we may write
\begin{align*}
\psi_{\lambda/\mu}(x) = x^{|\lambda| - |\mu|} f(1)^n \prod_{i=1}^n \frac{f(q^{\lambda_i} t^{n-i})}{f(q^{\mu_i} t^{n-i})}
\frac{1}{f(q^{\lambda_i - \mu_i})} \prod_{1 \le i < j \le n} \frac{f(q^{\mu_i - \mu_j} t^{j-i}) f(q^{\lambda_i - \lambda_j} t^{j-i-1})}{f(q^{\lambda_i - \mu_j} t^{j-i}) f(q^{\mu_i - \lambda_j} t^{j-i-1})}.
\end{align*}
If $\lambda_i = -\eps^{-1} \log(y_i)$ and $\mu_i = -\eps^{-1} \log(y_i')$, then
\begin{align} \label{eq:skew}
\begin{split}
P_{\lambda/\mu}(t^\alpha;q,t) & \sim \frac{\eps^{n(1-\theta)}}{\Gamma(\theta)^n} \prod_{i=1}^n \left( \frac{y_i}{y_i'} \right)^{\theta \alpha} \left( \frac{1 - y_i}{1 - y_i'} \right)^{1-\theta} \frac{1}{(1 - y_i/y_i')^{1-\theta}} \prod_{1 \le i < j \le n} \left[ \frac{(1 - y_i'/y_j')(1 - y_i/y_j)}{(1 - y_i/y_j')(1 - y_i'/y_j)} \right]^{1-\theta} \\
&= \frac{\eps^{n(1-\theta)}}{\Gamma(\theta)^n} \prod_{i=1}^n \frac{y_i^{\theta\alpha}}{(y_i')^{\theta(\alpha+1)-1}} \left( \frac{1 - y_i}{1 - y_i'} \right)^{1-\theta} \cdot \det \left[ \frac{1}{y_i' - y_j} \right]^{1 - \theta}
\end{split}
\end{align}
where the second line follows from the Cauchy determinant formula.

Combining \eqref{eq:pochhammer}, \eqref{eq:Q}, \eqref{eq:skew}, and $d\lambda_i \sim \eps^{-1} y_i^{-1} dy_i$ implies the convergence of \eqref{eq:prekernel} to
\[ \frac{1}{\Gamma(\theta)^n} \frac{\Gamma(\theta(\nu+n+1))}{\Gamma(\theta(\nu+1))} \cdot \left(\frac{\Delta(y)}{\Delta(x)}\right)^{\theta} \det \left[ \frac{1}{x_i - y_j} \right]^{1 - \theta} \prod_{i=1}^n \frac{y_i^{\theta(\nu+1)-1}}{x_i^{\theta(\nu+2)-1}} dy_i \]
which is exactly \eqref{TheJointDistributionTwoMatrices}.
\qed
\subsection{Proof of Theorem \ref{THEOREMProductTruncatedProcessTHETA}}\label{SectionSecondMainTheorem}
Set $X=T_1$ in the statement of Theorem  \ref{THEOREMMarkovKernel}, and apply Theorem \ref{THEOREMMarkovKernel}  to find the joint density
of squared singular values of $T_1$, $T_2T_1$, $T_3T_2T_1$, $\ldots$, $T_p\ldots T_1$. Taking into account that the distribution of the squared singular values of $T_1$ is given by equation
(\ref{DISTRIBUTIONONEMATRIX}) (with the normalization constant $Z_{n,\Jacobi}^{(\theta)}$ given by equation (\ref{ZJacobi})), we obtain formula (\ref{ProductTruncatedProcessTHETA}).
\qed
\section{Proof of Theorem \ref{THEOREMsvproduct}} \label{SECdixon}
We use an integration identity equivalent to one derived by Dixon \cite{Dixon}. A proof is provided in \cite[Exercise 4.2.2]{Forrester} based on unpublished work by Eric Rains.
\begin{prop} \label{prop:dixon}
Let $\alpha_0,\ldots,\alpha_n$, $\beta_0,\ldots,\beta_m$ have positive real parts, and suppose $\sum_{i=0}^n \alpha_i = \sum_{j=0}^m \beta_j$. Then
\begin{align*}
& \int_{R_n} \! dx_1 \cdots dx_n \, \prod_{1 \le i < j \le n} |x_j - x_i| \prod_{i=1}^n \left( \prod_{j=0}^n |a_j - x_i|^{\alpha_j - 1} \right) \left( \prod_{j=0}^m |b_j - x_i|^{-\beta_j} \right) \\
& \quad = \prod_{0 \le i < j \le n} (a_j - a_i)^{\alpha_i + \alpha_j - 1} \prod_{0 \le i < j \le m} (b_j - b_i)^{1 - \beta_i - \beta_j} \prod_{i=0}^n \prod_{j=0}^m |b_j - a_i|^{\alpha_i - \beta_j} \\
& \quad \quad \times \frac{\prod_{j=0}^n \Gamma(\alpha_j)}{\prod_{i=0}^m \Gamma(\beta_i)} \int_{R_m'} \! dx_1 \cdots dx_m \prod_{1 \le i < j \le m} |x_j - x_i| \prod_{i=1}^m \left( \prod_{j=0}^n |a_j - x_i|^{-\alpha_j} \right) \left( \prod_{j=0}^m |b_j - x_i|^{\beta_j - 1} \right)
\end{align*}
where $R_n$ and $R_m'$ denote the regions
\[ a_0 \le x_1 \le a_1 \le \cdots \le a_{n-1} \le x_n \le a_n, \quad b_0 \le x_1 \le b_1 \le \cdots \le b_{m-1} \le x_m \le b_m, \]
respectively.
\end{prop}

By sending $b_0 \to -\infty$, we obtain

\begin{cor} \label{cor:dixon}
Let $\alpha_0,\ldots,\alpha_n$, $\beta_0,\ldots,\beta_m$ have positive real parts, and suppose $\sum_{i=0}^n \alpha_i = \sum_{j=0}^m \beta_j$. Then
\begin{align*}
& \int_{R_n} \! dx_1 \cdots dx_n \, \prod_{1 \le i < j \le n} |x_j - x_i| \prod_{i=1}^n \left( \prod_{j=0}^n |a_j - x_i|^{\alpha_j - 1} \right) \left( \prod_{j=1}^m |b_j - x_i|^{-\beta_j} \right) \\
& \quad = \prod_{0 \le i < j \le n} (a_j - a_i)^{\alpha_i + \alpha_j - 1} \prod_{1 \le i < j \le m} (b_j - b_i)^{1 - \beta_i - \beta_j} \prod_{i=0}^n \prod_{j=1}^m |b_j - a_i|^{\alpha_i - \beta_j} \\
& \quad \quad \times \frac{\prod_{j=0}^n \Gamma(\alpha_j)}{\prod_{i=0}^m \Gamma(\beta_i)} \int_{R_m'} \! dx_1 \cdots dx_m \prod_{1 \le i < j \le m} |x_j - x_i| \prod_{i=1}^m \left( \prod_{j=0}^n |a_j - x_i|^{-\alpha_j} \right) \left( \prod_{j=1}^m |b_j - x_i|^{\beta_j - 1} \right)
\end{align*}
where $R_n$ and $R_m'$ denote the regions
\[ a_0 \le x_1 \le a_1 \le \cdots \le a_{n-1} \le x_n \le a_n, \quad -\infty < x_1 \le b_1 \le \cdots \le b_{m-1} \le x_m \le b_m, \]
respectively.
\end{cor}

To prove Theorem \ref{THEOREMsvproduct}, observe that the $p = 1$ case is clear. For our induction step, we apply the Markov kernel
\[ \frac{1}{\Gamma(\theta)^n} \frac{\Gamma(\theta(\nu_{p+1} + n + 1))}{\Gamma(\theta(\nu_{p+1} + 1))} \frac{\Delta\left( y^{(p+1)} \right)}{\Delta\left( y^{(p)} \right)^{2 \theta - 1}} \cdot \prod_{i,j=1}^n \left| y_i^{(p)} - y_j^{(p+1)} \right|^{\theta - 1} \prod_{i=1}^n \frac{\left(y_i^{(p+1)}\right)^{\theta(\nu_{p+1} + 1) - 1}}{\left( y_i^{(p)} \right)^{\theta(\nu_{p+1} + 2) - 1}} \]
to the density given by equation (\ref{svproduct}), integrating out $y^{(p)}$. We then apply Corollary \ref{cor:dixon} with
\[ \begin{array}{cccc}
a_0 = y_1^{(p+1)} & \cdots & a_{n-1} = y_n^{(p+1)} & a_n = 1  \\
\alpha_0 = \theta & \cdots & \alpha_{n-1} = \theta & \alpha_n = \theta(m_1 - 2n - \nu_1 + p)
\end{array},\]
$m=p$, and
\[ \begin{array}{ccccc}
b_0 = -\infty & b_1 = v_{p-1,1} & \cdots & b_{p-1} = v_{p-1,p-1} & b_p = 0 \\
\beta_0 = \theta(m_1 - n - \nu_{p+1}) & \beta_1 = \theta & \cdots & \beta_{p-1} = \theta & \beta_p = \theta(\nu_{p+1} - \nu_1 + 1)
\end{array} \]
where we require the condition $m_1 - n > \nu_{p+1}$ so that $\beta_0 > 0$. By our assumption, $\beta_p > 0$. Thus we obtain $I^{\theta,n,p+1}_{m_1;\nu_1,\ldots,\nu_{p+1}}$ as desired. \qed
\section{Proof of Theorem \ref{THEOREMDENSITYJACKREPRESENTATION}} \label{SECjack}
Assume
\[ \widetilde{A} = (t^{\nu+1},\ldots,t^{m_1 - n}, \ldots, t^{\nu_p+1},\ldots,t^{m_p - n}), \]
and suppose $T_1,T_2,\ldots$ be independent matrices where $T_i$ is a $(n+\nu_i)\times(n+\nu_{i-1})$ truncation of a Haar distributed matrices taken from $O(m_i)$, $U(m_i)$, or $\Sp(2m_i)$ for $\theta = 1/2,1,2$ respectively. Further assume that $n,m_i,\nu_i$ satisfy \eqref{Condition1} and \eqref{Condition2}.
\begin{prop} \label{PROPOSITIONJointLaw}
Let $\lambda$ be a random Young diagram whose distribution is defined by
$M_{\Macdonald}(\lambda;\widetilde{A},B)$. Assume that the Macdonald parameters $q$ and $t$ depend on $\epsilon > 0$, and are given by
\begin{equation}
q=q(\epsilon)=e^{-\epsilon},\;\; t=t(\epsilon)=\left(q(\epsilon)\right)^{\theta}=e^{-\theta\epsilon}.
\end{equation}
Let $x(\epsilon)=\left(x_1(\epsilon),\ldots,x_n(\epsilon)\right)$ be a random configuration associated with the Young diagram $\lambda$ as
$$
x_1(\epsilon)=e^{-\epsilon\lambda_1},\ldots,x_n(\epsilon)=e^{-\epsilon\lambda_n}.
$$
As $\epsilon\rightarrow 0+$, the distribution of $x_1(\epsilon)$, $\ldots$, $x_n(\epsilon)$ will coincide with that of squared singular values  of $T_p\cdots T_1$.
\end{prop}
\begin{proof}
By iterative applications of Lemma \ref{LEMzonal} and Lemma \ref{LEMswitch_product}, we have for any Young diagram $\kappa$ with length $\le n$
\[ \E \frac{J_\kappa(T_1^* \cdots T_p^* T_p \cdots T_1;\theta)}{J_\kappa(1^n;\theta)} = \prod_{i=1}^p \E \frac{J_\kappa(T_i^* T_i;\theta)}{J_\kappa(1^n;\theta)} \]
where we used the fact that the distributions of $T_i$ are invariant under right translation by Haar unitary matrices. Since the squared singular values of $T_i$ are Jacobi distributed, Lemma \ref{LEMselberg} implies
\begin{equation} \label{JACOBICauchy}
\E \frac{J_\kappa(T_1^* \cdots T_p^* T_p \cdots T_1;\theta)}{J_\kappa(1^n;\theta)} = \prod_{i=1}^p \prod_{j=1}^n \frac{\Gamma(\theta(\nu_i + n - j + 1) + \kappa_j)}{\Gamma(\theta(m_i - j + 1) + \kappa_j)} \frac{\Gamma(\theta(m_i - j + 1))}{\Gamma(\theta(\nu_i + n - j + 1))}.
\end{equation}
Since
\[ P_\kappa(u_1,\ldots,u_n) \to J_\kappa(u_1,\ldots,u_n) \]
uniformly over $u_1,\ldots,u_n \in [0,1]$ as $q,t \to 1$ with $t = q^\theta$, it suffices to show that
\begin{align*}
\E \frac{P_\kappa(q^{\lambda_1} t^{n-1},q^{\lambda_2} t^{n-2},\ldots,q^{\lambda_n})}{P_\kappa(B)} \to \E \frac{J_\kappa(T_1^* \cdots T_p^* T_p \cdots T_1;\theta)}{J_\kappa(1^n;\theta)}
\end{align*}
for every $\kappa$ with length $\le n$ where the expectation is with respect to $M_{\Macdonald}(\lambda; \widetilde{A},B)$ --- as in the proof of Proposition \ref{PROPOSITIONMarkovKernel}, this follows from the Stone-Weierstrass theorem. Indeed, we have
\begin{align*}
\E \frac{P_\kappa(q^{\lambda_1} t^{n-1},q^{\lambda_2} t^{n-2},\ldots,q^{\lambda_n})}{P_\kappa(B)} &= \E \frac{P_\lambda(q^{\kappa_1} t^{n-1},q^{\kappa_2} t^{n-2},\ldots,q^{\kappa_n})}{P_\lambda(B)} \\
&= \frac{1}{\Pi(\widetilde{A};B)} \sum_{\lambda \in \Y} Q_\lambda(\widetilde{A}) P_\lambda(q^{\kappa_1} t^{n-1},q^{\kappa_2} t^{n-2},\ldots,q^{\kappa_n}) \\
&= \frac{\Pi(\widetilde{A}; q^{\kappa_1} t^{n-1},q^{\kappa_2} t^{n-2},\ldots, q^{\kappa_n})}{\Pi(\widetilde{A};B)}
\end{align*}
where the first equality follows from \eqref{DUALITY} and the last equality from the Cauchy identity for Macdonald symmetric functions. The right hand side of the above converges as $\epsilon \to 0$ to the right hand side of \eqref{JACOBICauchy}.
\end{proof}
To complete the proof of Theorem \ref{THEOREMDENSITYJACKREPRESENTATION} let us
define $\lambda$ and $\widetilde{A}$ as in Proposition \ref{PROPOSITIONJointLaw} which states that
\[ x_1(\epsilon) = e^{-\epsilon \lambda_1}, \ldots, x_n(\epsilon) = e^{-\epsilon \lambda_n} \]
converges in distribution to the squared singular values of $T_p \cdots T_1$. Thus we establish the theorem by showing that
\[ \frac{1}{\Pi(\widetilde{A};B)} P_\lambda(\widetilde{A}) Q_\lambda(B) \]
converges to \eqref{JACKDensity} as $\epsilon \to 0$, where $\lambda_i = \epsilon^{-1} \log x_i(\epsilon)$, $q = e^{-\epsilon}$, and $t = q^\theta$. By our assumption
\[ \widetilde{A} = (q^{\mu_1} t^{M - 1}, q^{\mu_2} t^{M-2},\ldots, q^{\mu_M}) \]
up to reordering. By \eqref{DUALITY}, we have
\[ \frac{P_\lambda(\widetilde{A}) Q_\lambda(B)}{\Pi(\widetilde{A};B)} = \frac{1}{\Pi(\widetilde{A};B)} \frac{P_\mu(q^{\lambda_1}t^{M-1},q^{\lambda_2} t^{n-2}\ldots,q^{\lambda_n}t^{M-n}, t^{M-n-1},\ldots,1)}{P_\mu(t^{M-1},t^{M-2},\ldots,1)} P_\lambda(t^{M-1},t^{M-2},\ldots,1) Q_\lambda(B). \]

Let $\lambda(\epsilon)$ be a family of Young diagrams with $n$ rows. Assume that $\lambda(\epsilon)$ depends on a positive parameter $\epsilon$ in such a way that
$\epsilon\lambda_j(\epsilon)\rightarrow -\log x_j$ as $\epsilon\rightarrow 0+$, for some values $0<x_1\leq\ldots\leq x_n<1$.
Then
\[\underset{\epsilon\rightarrow 0+}{\lim}\left[\frac{P_\mu(q^{\lambda_1}t^{M-1},q^{\lambda_2} t^{n-2}\ldots,q^{\lambda_n}t^{M-n}, t^{M-n-1},\ldots,1)}{P_\mu(t^{M-1},t^{M-2},\ldots,1)}\right] =\frac{J_\mu(x_1,\ldots,x_n,1^{M-n};\theta)}{J_\mu(1^M;\theta)} \]
In addition, we use several asymptotics from the proof of Theorem \ref{THEOREMMarkovKernel}. From \eqref{MACDONALDFormula} and \eqref{ASYMPTOTICSt}, under the same assumptions on the family $\lambda(\epsilon)$ of Young diagrams as above we have
\[ P_\lambda(t^{M-1},t^{M-2},\ldots,1) \sim \epsilon^{-\theta ((M-n)n + n(n-1)/2)} \prod_{\substack{i < j \\ 1 \le i \le n \\ 1 \le j \le M}} \frac{\Gamma(\theta(j-i))}{\Gamma(\theta(j - i + 1))} \prod_{1 \le i < j \le n} (x_j - x_i)^\theta \prod_{i=1}^n (1 - x_i)^{\theta(M-n)}. \]
By \eqref{eq:pochhammer}, we have
\[ \frac{1}{\Pi(\widetilde{A};B)} = \prod_{i=1}^n \prod_{r=1}^p \prod_{j=\nu_r + 1}^{m_r - n} \frac{(t^{j+i-1};q)_\infty}{(t^{j + i};q)_\infty} \sim \epsilon^{\theta Mn} \prod_{i=1}^n \prod_{r=1}^p \prod_{j=\nu_r + 1}^{m_r - n} \frac{\Gamma(\theta(j + i))}{\Gamma(\theta(j + i - 1))}, \]
and \eqref{eq:Q} gives the relevant asymptotics of $Q_\lambda(1,\ldots,t^{n-1})$.
Combining the asymptotics above and the fact that $d\lambda_i \sim \epsilon^{-1} x_i^{-1} dx_i$, we obtain the desired result after simplifying the Gamma factors. \qed
\section{The derivation of formulas for the singular values of the product of two truncated symplectic matrices} \label{SECpfaffian}
In this Section we provide the derivations of different formulas stated in Section \ref{SectionCorrelationFunctionsStatements}. These formulae describe the distribution of singular values of the product of two truncated symplectic matrices. We start from the proof of Proposition \ref{PROPOSITIONDENSITYDETERMINANT} which gives the density of squared singular values as a determinant.
\subsection{Proof of Proposition \ref{PROPOSITIONDENSITYDETERMINANT}}\label{FirstPropositionOnCorrelationFunctions}
Let $Y=T_2T_1$, where $T_1$ is the $\left(n+\nu_1\right)\times n$ truncation of a Haar distributed symplectic matrix $S_1$ of size $m_1\times m_1$, and $T_2$ is the $\left(n+\nu_2\right)\times\left(n+\nu_1\right)$ truncation of a Haar distributed symplectic matrix $S_2$ of size $m_2\times m_2$. The distribution of the squared singular values $\left(x_1,\ldots,x_n\right)$ of $T_1$ is the Jacobi ensemble with the parameters $\theta=2$, $\nu=\nu_1$, and $m=m_1$ given by equation (\ref{DISTRIBUTIONONEMATRIX}). Clearly, the density of the squared singular values
$\left(y_1,\ldots,y_n\right)$ of the product matrix $T_2T_1$ can be obtained using  Theorem \ref{THEOREMMarkovKernel} (with $\theta=2$). Namely, the probability measure defined by equation (\ref{TheJointDistributionTwoMatrices})  leads to a Markov kernel for the product matrix process formed by truncated symplectic matrices. It is not hard to see
that this Markov kernel (for $\theta=2$)  can be written as
\begin{equation}\label{MK1}
\frac{\Gamma\left(2\left(\nu_2+n+1\right)\right)}{\Gamma\left(2\left(\nu_2+1\right)\right)}\prod\limits_{i,j=1}^n\left(x_i-y_j\right)\frac{\triangle(y)}{\triangle(x)^3}
\det\left[\frac{y_j^{2\nu_2+1}}{x_i^{2\nu_2+3}}\mathbf{1}[x_i\geq y_j]\right]_{i,j=1}^n,
\end{equation}
where $0<x_1<\ldots<x_n$. Rewriting the double product in terms of Vandermonde determinants, we see that expression (\ref{MK1}) can be also rewritten as
\begin{equation}\label{MK2}
\frac{\Gamma\left(2\left(\nu_2+n+1\right)\right)}{\Gamma\left(2\left(\nu_2+1\right)\right)}\frac{\triangle(y,x)}{\triangle(x)^4}
\det\left[\frac{y_j^{2\nu_2+1}}{x_i^{2\nu_2+3}}\mathbf{1}[x_i\geq y_j]\right]_{i,j=1}^n.
\end{equation}
We apply this kernel to the density of squared singular values of $T_1$ given by equation (\ref{DISTRIBUTIONONEMATRIX}) with the parameters $\theta=2$, $\nu=\nu_1$,  $m=m_1$, and obtain
\begin{equation}
\const\prod\limits_{j=1}^ny_j^{2\nu_2+1}\int\limits_0^{\infty}\ldots\int\limits_{0}^{\infty}dx_1\ldots dx_n
\det\left[\varphi_j\left(y_1\right),\ldots,\varphi_j\left(y_n\right),\varphi_j\left(x_1\right),\ldots,\varphi_j\left(x_n\right)\right]_{1\leq j\leq 2n}
\det\left[\psi_i\left(x_j\right)\right]_{i,j=1}^n,
\nonumber
\end{equation}
where
$$
\varphi_j(x)=x^{j-1},\; 1\leq j\leq 2n;\;\; \psi_i(x)=x^{2\left(\nu_1-\nu_2-1\right)}\left(1-x\right)^{2\left(m_1-2n-\nu_1\right)+1}1\left[y_i\leq x\leq 1\right],\;\; 1\leq i\leq n.
$$
Representing the determinants as sums over permutations we can rewrite the expression above as a single determinant, namely as
\begin{equation}\label{SingleDeterminant}
\const\prod\limits_{j=1}^ny_j^{2\nu_2+1}\det\left[\begin{array}{ccc}
                                                    \varphi_1(y) & \ldots & \varphi_{2n}(y_1) \\
                                                    \vdots &  & \vdots \\
                                                    \varphi_1(y_n) & \ldots & \varphi_{2n}(y_n) \\
                                                    \int_0^{\infty}\psi_1(t)\varphi_1(t)dt & \ldots & \int_0^{\infty}\psi_1(t)\varphi_{2n}(t)dt \\
                                                    \vdots &  & \vdots \\
                                                     \int_0^{\infty}\psi_n(t)\varphi_1(t)dt & \ldots &  \int_0^{\infty}\psi_n(t)\varphi_{2n}(t)dt
                                                  \end{array}
\right].
\end{equation}
Changing variables $t=y_i/\tau$, we have
$$
\int_0^{\infty}\psi_i(t)\varphi_j(t)dt=\int\limits_{y_i}^1\left(\frac{y_i}{\tau}\right)^{2\left(\nu_1-\nu_2-1\right)+j}\left(1-\frac{y_i}{\tau}\right)^{2\left(m_1-2n-\nu_1\right)+1}
\frac{d\tau}{\tau}.
$$
Recall that
$$
\frac{x^b(1-x)^{a-b-1}}{\Gamma(a-b)}\mathbf{1}_{[0,1]}(x)=G_{1,1}^{1,0}\left(\begin{array}{c}
                                                          a \\
                                                          b
                                                        \end{array}
\biggr\vert x\right),
$$
see, for example, Luke \cite{Luke}, Section 6.6, equation (3).
Therefore,
$$
\left(\frac{y_i}{\tau}\right)^{2\nu_1+j-1}\left(1-\frac{y_i}{\tau}\right)^{2\left(m_1-2n-\nu_1\right)+1}
\mathbf{1}_{[0,1]}\left(\frac{y_i}{\tau}\right)=\Gamma\left(2\left(m_1-2n-\nu_1+1\right)\right)
G_{1,1}^{1,0}\left(\begin{array}{c}
                                                          2\left(m_1-2n\right)+j+1 \\
                                                          j+\nu_1-1
                                                        \end{array}
\biggr\vert \frac{y_i}{\tau}\right),
$$
which gives
$$
y_i^{2\nu_2+1}\int_0^{\infty}\psi_i(t)\varphi_j(t)dt=\Gamma\left(2\left(m_1-2n-\nu_1+1\right)\right)\int\limits_0^1\frac{d\tau}{\tau}\tau^{2\nu_2+1}
G_{1,1}^{1,0}\left(\begin{array}{c}
                                                          2\left(m_1-2n\right)+j+1 \\
                                                          j+\nu_1-1
                                                        \end{array}
\biggr\vert \frac{y_i}{\tau}\right).
$$
The following formula holds true
$$
\int\limits_0^1x^{\beta}(1-x)^{\alpha-\beta-1}G_{1,1}^{1,0}\left(\begin{array}{c}
                                                          a \\
                                                          b
                                                        \end{array}
\biggr\vert \frac{y}{x}\right)\frac{dx}{x}=\Gamma(\alpha-\beta)G_{2,2}^{2,0}\left(\begin{array}{cc}
                                                                \alpha & a \\
                                                                \beta & b
                                                              \end{array}
\biggr\vert y\right),
$$
see, for example,  Luke \cite{Luke}, Section 5.6, equation (6).
Taking this into account we see that the integrals in the determinant in expression (\ref{SingleDeterminant}) can be rewritten in terms of the corresponding Meijer $G$-functions. This gives the formula in the statement
of Proposition \ref{PROPOSITIONDENSITYDETERMINANT}.
\qed
\subsection{A matrix representation for the correlation kernel. Proof of Proposition \ref{PROPOSITIONK}}\label{SectionSecondPropositionOnCorrelationFunctions}
\subsubsection{The formula for the correlation kernel for a general symplectic-type ensemble}
Proposition \ref{PROPOSITIONDENSITYDETERMINANT} implies that squared singular values $\left(x_1,\ldots,x_n\right)$ of $X=T_2T_1$ form a symplectic-type ensemble in the sense of Section \ref{SectionCorrelationFunctionsStatements}.
Indeed, the density $P_{n,\Product}^{(2)}\left(x_1,\ldots,x_n\right)$ of squared singular values $\left(x_1,\ldots,x_n\right)$ of $X=T_2T_1$ can be written as
$$
P_{n,\Product}^{(2)}\left(x_1,\ldots,x_n\right)=\const\det\left(\phi_j(x_k)\;\psi_j(x_k)\right)_{1\leq k\leq n,\; 0\leq j\leq 2n-1},
$$
where
$$
\phi_j(x)=x^j,\;\;\;
\psi_j(x)=G_{2,2}^{2,0}\left(\begin{array}{cc}
                                                       2\nu_2+2 & 2\left(m_1-2n+1\right)+j \\
                                                       2\nu_2+1 & 2\nu_1+j
                                                     \end{array}
\biggr\vert x\right),
$$
and $0\leq j\leq 2n-1$. The next Proposition gives a $2\times 2$ matrix representation for the correlation kernel of a general symplectic-type ensemble.
\begin{prop}\label{GeneralPropositionSymplecticTypeEnsemble}
Consider a symplectic-type ensemble defined by probability measure (\ref{SymplecticTypeEnsemble}), where $\phi_j(x)$, $\psi_j(x)$ are certain functions,
$0\leq j\leq 2n-1$, and $Z_n$ is the normalizing constant. The correlation kernel $\mathbb{K}_n(x,y)$ of this ensemble defined by equation (\ref{CorrelationKernelSymplecticTypeEnsemble})
can be written as
\begin{equation}
\mathbb{K}_n(x,y)=\left(\begin{array}{cc}
               \sum\limits_{k,l=0}^{2n-1}\psi_k(x)q_{k,l}\phi_l(y)  & -\sum\limits_{k,l=0}^{2n-1}\psi_k(x)q_{k,l}\psi_l(y) \\
                \sum\limits_{k,l=0}^{2n-1}\phi_k(x)q_{k,l}\phi_l(y) & -\sum\limits_{k,l=0}^{2n-1}\phi_k(x)q_{k,l}\psi_l(y)
              \end{array}
\right),
\end{equation}
where $Q=\left(q_{i,j}\right)_{i,j=0}^{2n-1}$ is the inverse of $C=\left(c_{i,j}\right)_{i,j=0}^{2n-1}$ defined by
$$
c_{i,j}=\int\left(\phi_i(x)\psi_j(x)-\phi_j(x)\psi_i(x)\right)dx.
$$
\end{prop}
\begin{proof} See Tracy and Widom \cite{TracyWidom}.
\end{proof}
\subsubsection{Proof of Proposition \ref{PROPOSITIONK}}
The distribution of squared singular values for the product of two truncated symplectic matrices is the symplectic-type ensemble  (\ref{SymplecticTypeEnsemble})
with
\begin{equation}
\phi_j(x)=x^j,
\psi_j(x)=G_{2,2}^{2,0}\left(\begin{array}{cc}
                                                                                                 2\nu_2+2 & 2\left(m_1-2n+1\right)+j \\
                                                                                                 2\nu_2+1 & 2\nu_1+j
                                                                                               \end{array}
\biggr|x\right),
\end{equation}
where $0\leq j\leq 2n-1$. The matrix $C=\left(c_{i,j}\right)_{i,j=0}^{2n-1}$  in Proposition \ref{GeneralPropositionSymplecticTypeEnsemble} can be computed explicitly, the computation gives equation (\ref{cproduct}).
\qed
\subsection{A general formula for the inverse of a skew-symmetric Hankel-type matrix. Proof of Proposition \ref{PROPOSITIONInverse}}\label{SectionProofInverse}
In order to obtain explicit formulae for the matrix entries of the kernel $\mathbb{K}_{n,\Product}(x,y)$ of Proposition \ref{PROPOSITIONK}
we need to find explicitly the inverse of $C^{\Product}=\left(c_{i,j}^{\Product}\right)_{i,j=0}^{2n-1}$ defined by equation (\ref{cproduct}).
The matrix $C^{\Product}$ can be understood as a \textit{skew-symmetric Hankel type} matrix, see Definition \ref{DEFINITIONHANKELTYPEMATRIX}
below. In this Section we first derive a general formula for the inverse of a skew-symmetric Hankel type matrix, see Proposition \ref{PropGeneralInverse}
below. Then we apply Proposition \ref{PropGeneralInverse} to derive the formulae stated in Proposition \ref{PROPOSITIONInverse}.
\begin{defn}\label{DEFINITIONHANKELTYPEMATRIX} Let $\mu$ be a positive measure on $\R$ with finite moments,
\begin{equation}\label{H1}
h_k=\int\limits_{\R}x^k\mu(dx),\;\; k=0,1,2,\ldots.
\end{equation}
We will refer to the skew-symmetric matrix $H$  of size $2n\times 2n$ defined by
\begin{equation}\label{H2}
H=\left((j-i)h_{i+j-1}\right)_{i,j=0}^{2n-1}
\end{equation}
as to a \textit{skew-symmetric Hankel type} matrix.
\end{defn}
Here we find a general formula for the inverse of $H$.
\begin{prop}\label{PropGeneralInverse}
Let $\langle.,.\rangle_{\mu}$ denote the skew inner product defined in terms of $\mu$,
\begin{equation}
\langle f,g\rangle_{\mu}=\frac{1}{2}\int\limits_{\R}\left(f(x)g'(x)-g(x)f'(x)\right)\mu(dx).
\end{equation}
Let $\left\{q_{2k}(x),q_{2k+1}(x)\right\}_{k=0}^{n-1}$ be a system of skew orthogonal polynomials  satisfying the following condition
\begin{equation}\label{SkewOrthogonalCondition}
\langle q_i,q_j\rangle_{\mu}=\left\{
                               \begin{array}{ll}
                                 r_k, & \hbox{if}\; i=2k, j=2k+1\; \hbox{for some}\; k\in \{0,1,\ldots,n-1\}, \\
                                -r_k, & \hbox{if}\; i=2k+1, j=2k\; \hbox{for some}\; k\in \{0,1,\ldots,n-1\},\\
                                 0, & \hbox{otherwise.}
                               \end{array}
                             \right.
\end{equation}
Then the $2n\times 2n$ skew-symmetric matrix  $Q$ be defined by
\begin{equation}\label{Q1}
Q=\left(\frac{1}{2}\left(q_{j,i}-q_{i,j}\right)\right)_{i,j=0}^{2n-1},
\end{equation}
where
\begin{equation}\label{FormulaForQij}
q_{i,j}=\sum\limits_{k=0}^{2n-1}\frac{1}{r_k}\left(\frac{1}{i!}\frac{d^i}{dx^i}q_{2k}(x)\right)\biggl|_{x=0}
\left(\frac{1}{j!}\frac{d^j}{dy^j}q_{2k+1}(y)\right)\biggl|_{y=0},
\end{equation}
is the inverse of the skew-symmetric Hankel-type matrix $H$ defined by equations (\ref{H1}) and (\ref{H2}).
\end{prop}
\begin{proof}
It can be checked that the kernel
\begin{equation}
S_n(x,y)=\sum\limits_{k=0}^{n-1}\frac{1}{r_k}\left(q_{2k}(x)q_{2k+1}(y)-q_{2k+1}(x)q_{2k}(y)\right)
\end{equation}
has the reproducing property
\begin{equation}\label{Reproducing}
y^k=\langle S_n(x,y),x^k\rangle_{\mu}
\end{equation}
Write
\begin{equation}\label{Q2}
\sum\limits_{k=0}^{n-1}\frac{1}{r_k}\left(q_{2k}(x)q_{2k+1}(y)-q_{2k+1}(x)q_{2k}(y)\right)=\sum\limits_{i,j=0}^{2n-1}q_{i,j}\left(x^iy^j-x^jy^i\right).
\end{equation}
The equation above defines the coefficients $q_{i,j}$. Taking into account this equation we see that
the inner product $\langle S_n(x,y),x^k\rangle_{\mu}$  can be rewritten as
\begin{equation}
\langle S_n(x,y),x^k\rangle_{\mu}=\frac{1}{2}\sum\limits_{i,j=0}^{2n-1}\left(q_{i,j}-q_{j,i}\right)(k-i)h_{k+i-1}y^i
\end{equation}
which implies
\begin{equation}
\frac{1}{2}\sum\limits_{i=0}^{2n-1}(i-k)h_{k+i-1}\left(q_{j,i}-q_{i,j}\right)=\delta_{k,j}.
\end{equation}
Therefore, the matrix $Q$ defined by equations (\ref{Q1}) and (\ref{Q2}) is the inverse of the skew-symmetric Hankel-type matrix $H$ defined by equations (\ref{H1}) and (\ref{H2}). Moreover, it is not hard to see that $q_{i,j}$ (defined by (\ref{Q2})) can be determined from   equation (\ref{FormulaForQij})
\end{proof}
If  $d\mu(x)=(1-x)^{a+1}(1+x)^{b+1}dx$  (this is a measure on $[-1,1]$), then  we have
\begin{equation}\label{hkJacobi}
h_k=\int\limits_{-1}^{1}(x-1)^kd\mu(x)=\frac{2^{a+b+k+3}}{(-1)^k}\frac{\Gamma(a+k+2)\Gamma(b+2)}{\Gamma(a+b+k+4)},
\end{equation}
where $k=0,1,\ldots$.
\begin{prop}\label{PropositionGIJ}Let $\{q_k^{\Jacobi}\}$ be the family of the skew-orthogonal polynomials  with respect to $\langle.,.\rangle_{\mu}$, where $\mu$ is defined by
equation (\ref{hkJacobi}), and let
\begin{equation}
S_n^{\Jacobi}(x,y)=\sum\limits_{k=0}^{n-1}\frac{1}{r_k^{\Jacobi}}\left(q_{2k}^{\Jacobi}(x)q_{2k+1}^{\Jacobi}(y)-q_{2k+1}^{\Jacobi}(x)q_{2k}^{\Jacobi}(y)\right)
\end{equation}
be the corresponding reproducing kernel. Define the coefficients $q_{i,j}^{\Jacobi}$  from the expansion
\begin{equation}
S_n^{\Jacobi}(x,y)=\sum\limits_{i,j=0}^{2n-1}q_{i,j}^{\Jacobi}\left((x-1)^i(y-1)^j-(x-1)^j(y-1)^i\right).
\end{equation}
Then the matrix $q^{\Jacobi}=\left(\frac{1}{2}\left(q_{j,i}^{\Jacobi}-q_{i,j}^{\Jacobi}\right)\right)_{i,j=0}^{2n-1}$ is the inverse of
$$
C^{\Jacobi}=\left((j-i)\frac{2^{a+b+i+j+2}}{(-1)^{i+j-1}}\frac{\Gamma(a+i+j+1)\Gamma(b+2)}{\Gamma(a+b+i+j+3)}\right)_{i,j=0}^{2n-1}.
$$
The coefficients $q_{j,i}^{\Jacobi}$ can be written as
\begin{equation}
q_{i,j}^{\Jacobi}=\sum\limits_{k=0}^{n-1}\frac{1}{r_k^{\Jacobi}}\left(\frac{1}{i!}\frac{d^i}{dx^i}q_{2k}^{\Jacobi}(x)\right)\biggl|_{x=1}
\left(\frac{1}{j!}\frac{d^j}{dy^j}q_{2k+1}^{\Jacobi}(y)\right)\biggl|_{y=1}.
\end{equation}
\end{prop}
\begin{proof}The proof of Proposition \ref{PropositionGIJ} is very similar to that of  Proposition \ref{PropGeneralInverse}. The only difference is that we use
$$
(y-1)^k=\langle S_n^{\Jacobi}(x,y),(x-1)^k\rangle_{\mu}
$$
instead of (\ref{Reproducing}).
\end{proof}
\begin{prop}\label{PROPOSITIONQJ} We have
\begin{equation}
\begin{split}
&q_{i,j}^{\Jacobi}=\frac{1}{2^{a+b+i+j+1}}\sum\limits_{k=0}^{n-1}\sum\limits_{l=0}^k\frac{2^{4k}}{2^{4l}}
\frac{(a+b+4l+1)(a+b+4k+3)\Gamma(a+1+2l)\Gamma(a+2+2k)}{\Gamma(l+1)\Gamma\left(\frac{a}{2}+\frac{b}{2}+l+1\right)
\Gamma\left(\frac{a}{2}+l+1\right)\Gamma\left(\frac{b}{2}+l+1\right)}\\
&\times\Theta\left(k+1\right)\Theta\left(\frac{a}{2}+k+1\right)\Theta\left(\frac{b}{2}+k+1\right)
\Theta\left(\frac{a}{2}+\frac{b}{2}+k+1\right)\\
&\times\frac{\Gamma(a+b+2l+i+1)\Gamma(a+b+2k+j+2)}{\Gamma(a+i+1)\Gamma(a+j+1)}\left(\begin{array}{c}
                                                                                          2l \\
                                                                                          i
                                                                                        \end{array}
\right)
\left(\begin{array}{c}
                                                                                          2k+1 \\
                                                                                          j
                                                                                        \end{array}
\right),
\end{split}
\end{equation}
where $\Theta(x)=\Gamma(x)/\Gamma(2x)$.
\end{prop}
\begin{proof}
The skew-orthogonal polynomials $\left\{q_{k}^{\Jacobi}\right\}_{k=0}^{\infty}$ are given by
\begin{equation}
q_{2k}^{\Jacobi}(x)=2^{6k}k!\sum\limits_{l=0}^{k}\frac{\Gamma\left(\frac{a}{2}+\frac{b}{2}+k+1\right)\Gamma\left(\frac{a}{2}+k+1\right)
\Gamma\left(\frac{b}{2}+k+1\right)}{\Gamma\left(\frac{a}{2}+\frac{b}{2}+l+1\right)\Gamma\left(\frac{a}{2}+l+1\right)
\Gamma\left(\frac{b}{2}+l+1\right)}
\frac{\Gamma(a+b+4l+2)}{\Gamma(a+b+4k+2)}\frac{p_{2l}^{\Jacobi}(x)}{l!2^{6l}},
\end{equation}
and
\begin{equation}
q_{2k+1}^{\Jacobi}(x)=p_{2k+1}^{\Jacobi}(x),
\end{equation}
see Adler, Forrester, Nagao, and van Moerbeke \cite{AdlerForresterNagaoMoerbeke}.
Here $\left\{p_{k}^{\Jacobi}(x)\right\}_{k=0}^{\infty}$ are polynomials defined by
\begin{equation}
p_{k}^{\Jacobi}(x)=2^kk!\frac{\Gamma(a+b+k+1)}{\Gamma(a+b+2k+1)}P_k^{(a,b)}(x),
\end{equation}
where $P_k^{(a,b)}(x)$ are the Jacobi polynomials. Since
\begin{equation}
\frac{d^k}{dx^k}P_n^{(a,b)}(x)=\frac{(n+a+b+1)_k}{2^k}P_{n-k}^{(a+k,b+k)}(x),
\end{equation}
and
\begin{equation}
P_n^{(a,b)}(1)=\frac{(a+1)_n}{n!},
\end{equation}
we obtain
\begin{equation}
\frac{1}{i!}\frac{d^i}{dx^i}P_k^{(a,b)}(x)\biggl|_{x=1}=\frac{\Gamma(a+1+k)\Gamma(a+b+1+k+i)}{2^ii!(k-i)!\Gamma(a+i+k)\Gamma(a+b+1+k)},
\end{equation}
which gives
\begin{equation}
\begin{split}
&\frac{1}{i!}\frac{d^i}{dx^i}q^{\Jacobi}_{2k}(x)\biggl|_{x=1}=\sum\limits_{l=0}^k\frac{2^{6k-4l-i}k!}{l!}
\frac{\Gamma\left(\frac{a}{2}+\frac{b}{2}+k+1\right)
\Gamma\left(\frac{a}{2}+k+1\right)\Gamma\left(\frac{b}{2}+k+1\right)}{\Gamma\left(\frac{a}{2}+\frac{b}{2}+l+1\right)
\Gamma\left(\frac{a}{2}+l+1\right)\Gamma\left(\frac{b}{2}+l+1\right)}\\
&\times\frac{\Gamma(a+1+2l)\Gamma(a+b+4l+2)\Gamma(a+b+2l+i+1)}{\Gamma(a+b+4l+1)\Gamma(a+b+4k+2)\Gamma(a+i+1)}\left(\begin{array}{c}
                                                                                                                             2l \\
                                                                                                                             i
                                                                                                                           \end{array}
\right),
\end{split}
\end{equation}
and
\begin{equation}
\frac{1}{j!}\frac{d^j}{dy^j}q^{\Jacobi}_{2k+1}(y)\biggl|_{y=1}=2^{2k+1-j}\left(
\begin{array}{c}
                                                                                                                             2k+1 \\
                                                                                                                             j
                                                                                                                           \end{array}\right)
\frac{\Gamma(a+2k+2)\Gamma(a+b+2k+j+2)}{\Gamma(a+j+1)\Gamma(a+b+4k+3)}
\end{equation}
Taking into account that
\begin{equation}
\frac{1}{r_k^{\Jacobi}}=\frac{\Gamma(a+b+4k+2)\Gamma(a+b+4k+4)}{2^{a+b+4k+2}(2k+1)!\Gamma(a+2k+2)\Gamma(b+2k+2)\Gamma(a+b+2k+2)},
\end{equation}
we obtain the formula in the statement of the Proposition.
\end{proof}
\textit{Proof of Proposition }\ref{PROPOSITIONInverse}. Proposition \ref{PROPOSITIONInverse} is an immediate Corollary of Proposition \ref{PropositionGIJ}
and Proposition \ref{PROPOSITIONQJ}. \qed
\subsection{Proof of Theorem \ref{THEOREMCORRELATIONKERNELFINAL}}\label{SectionFinalTheoremProof}
Now we are ready to derive explicit formulae for the matrix entries of the correlation kernel $\mathbb{K}_{n,\Product}(x,y)$ stated in  Theorem \ref{THEOREMCORRELATIONKERNELFINAL}. We use equations (\ref{K11}), (\ref{K12}), (\ref{K21}), and (\ref{K22}). All these formulae involve the coefficients
$q_{k,l}^{\Product}$ which can be understood as matrix entries of the inverse of $C^{\Product}=\left(c_{i,j}^{\Product}\right)_{i,j=0}^{2n-1}$
defined by equation (\ref{cproduct}). Proposition \ref{PROPOSITIONInverse} can be used to find the inverse of $C^{\Product}$ explicitly. As a result, we obtain
the desired formulae for the matrix entries of the correlation kernel $\mathbb{K}_{n,\Product}(x,y)$.
\qed

\section{Extension to Arbitrary $\beta > 0$} \label{SECarbitrarybeta}

In this section, we describe a generalization of our model to arbitrary $\theta > 0$. We write $\beta = 2\theta$, to indicate the connection with $\beta$-ensembles from random matrix theory. Further details about these $\beta$-deformed models can be found in \cite{GorinMarcus} and \cite{Ahn}.

We begin by recalling the $\beta$-Jacobi ensemble with parameters $m,\nu$. Suppose $T$ is an $(n+\nu)\times n$-truncation of a Haar distributed matrix $S$ taken from $U(m)$, $O(m)$, $\Sp(2m)$ for $\beta = 1,2,4$ respectively. Recall that (see Proposition \ref{PROPOSITIONJacobiEnsembles}) the distribution on $[0,1]^n$ defined by the density
\begin{equation} \label{DISTRIBUTIONJacobi}
\begin{split}
&P_{n,\Jacobi}^{(\beta/2)}\left(x_1,\ldots,x_n\right)dx_1\ldots dx_n\\
&=\frac{1}{Z_{n,\Jacobi}^{(\beta/2)}}\prod\limits_{1\leq j<k\leq n}\left|x_j-x_k\right|^{\beta}\prod\limits_{j=1}^n
\left(x_j\right)^{(\beta/2)\left(\nu+1\right)-1}\left(1-x_j\right)^{(\beta/2)(m-2n-\nu+1)-1}dx_1\ldots dx_n,
\end{split}
\end{equation}
where $Z_{n,\Jacobi}^{(\beta/2)}$ is given by \eqref{ZJacobi}, is the distribution of the eigenvalues $(x_1,\ldots,x_n)$ of $T^*T$. Here, we assume $m \ge 2n + \nu$. Though we require $\beta \in \{1,2,4\}$ in order to interpret $(x_1,\ldots,x_n)$ as eigenvalues of some invariant matrix ensemble, there is no obstruction in defining a random variable $(x_1,\ldots,x_n)$ on $[0,1]^n$ with density \eqref{DISTRIBUTIONJacobi} for arbitrary $\beta > 0$. The resulting distribution is referred to as the $\beta$-Jacobi ensemble and is among the three classical $\beta$-ensembles alongside the $\beta$-Laguerre and $\beta$-Hermite ensembles. For an introduction to $\beta$-ensembles and further literature, we refer the reader to \cite[Chapter 20]{AkemannBaikDiFrancesco} and references therein.

We can go further and consider a $\beta > 0$ extension of the Markov transition kernel from Theorem \ref{THEOREMMarkovKernel}. Recall that (see Theorem \ref{THEOREMMarkovKernel}) if we assume $m = n + \nu + 1$, then for a deterministic $n\times n$ matrix $X$ with squared singular values $x = (x_1<\cdots<x_n) \in [0,1]^n$, we have that the joint density of the ordered squared singular values $y = (y_1<\cdots<y_n)$ of $TX$ is
\begin{align*}
& p_{n,\nu}^{(\beta/2)}(y|x) \\
& = \frac{1}{\left(\Gamma(\beta/2)\right)^n}\frac{\Gamma\left((\beta/2)\left(\nu+n+1\right)\right)}{\Gamma\left((\beta/2)(\nu+1)\right)}
\left(\frac{\triangle(y)}{\triangle(x)}\right)^{\beta/2}\left(\det\left[\frac{1}{x_i-y_j}\right]_{i,j=1}^n\right)^{1-\beta/2}
\prod\limits_{i=1}^n\frac{y_i^{(\beta/2)(\nu+1)-1}}{x_i^{(\beta/2)(\nu+2)-1}}dy_i
\end{align*}
supported in
\begin{equation} \label{Interlacing}
0 \le y_1 \le x_1 \le \cdots \le y_n \le x_n.
\end{equation}
By the Cauchy determinant formula, we may reexpress this density as
\begin{align} \label{BetaMarkov}
\begin{split}
p_{n,\nu}^{(\beta/2)}(y|x) = \frac{1}{(\Gamma(\beta/2))^n} \frac{\Gamma((\beta/2)(\nu+n+1))}{\Gamma((\beta/2)(\nu+1))} \frac{\triangle(y)}{\triangle(x)^{\beta - 1}} \prod_{i,j=1}^n |x_i - y_j|^{\beta/2 - 1} \prod_{i=1}^n \frac{y_i^{(\beta/2)(\nu+1) - 1}}{x_i^{(\beta/2)(\nu+2) - 1}} dy_i.
\end{split}
\end{align}
While the interpretation of \eqref{BetaMarkov} as the density of squared singular values of a product of a truncated $O(m)$, $U(m)$, and $\Sp(2m)$ matrix with a deterministic matrix only makes sense for $\beta = 1,2,4$, we can still make sense of a random variable $(y_1,\ldots,y_n)$ supported in \eqref{Interlacing} with density \eqref{BetaMarkov} for arbitrary $\beta > 0$, just as with the $\beta$-Jacobi ensemble.

From another perspective, the density \eqref{BetaMarkov} can also be obtained by scaling limit of the transition $K_{\Markov}(\lambda,\mu;A,B)$ as in Proposition \ref{PROPOSITIONMarkovKernel} by taking $\theta = \beta/2$. Indeed, the proof of Proposition \ref{PROPOSITIONMarkovKernel} follows verbatim. By viewing \eqref{BetaMarkov} itself as a transition kernel, we can iterate the kernel and define a $\beta > 0$ generalization of the matrix product process:

\begin{defn}
Fix $\beta > 0$, $\nu_k > 1$ for $k = 2,3,\ldots$, and $x^1 = (x_1^1 < \cdots < x_n^1) \subset [0,1]^n$. Define the \emph{$\beta$-Jacobi product process with initial state $x^1$ and parameters $(\nu_k)_{k=2}^\infty$} to be the Markov process in discrete time $(x^k)_{k=1}^\infty$ where the Markov transition kernel from $x^{k-1}$ to $x^k$ is given by $p_{n,\nu_k}^{(\beta/2)}(x^k|x^{k-1})$ as in \eqref{BetaMarkov}.
\end{defn}

The limit shape and fluctuations of the $\beta$-Jacobi product process, where the initial state was distributed as a $\beta$-Jacobi ensemble, were studied in \cite{Ahn}.

We can further interpret the density \eqref{BetaMarkov} as coming from an operation on vectors which is some generalization of products of $O(m)$, $U(m)$, and $\Sp(2m)$ invariant matrices at the level of singular values. In particular, many properties from the $\beta = 1,2,4$ cases generalize readily to arbitrary $\beta > 0$. For details, we refer the reader to \cite{Ahn}. We indicate one of these properties here.

By Lemma \ref{LEMswitch_product}, the squared singular values of $TX$ are distributed as the squared singular values of $(T^*T)^{1/2} X$. Since $m = n + \nu + 1$, the condition $m \ge 2n + \nu$ is not satisfied. In particular, $n - 1$ of the eigenvalues of $T^*T$ must be $1$, and only one eigenvalue is deterministic. Nonetheless, it is known (see e.g. Forrester \cite{Forrester}, Section 3.8.3.) that if $(u_1,1,\ldots,1)$ are the eigenvalues of $T^*T$, then the density of $u_1$ is $P_{1,\Jacobi}^{(\beta/2)}$ with parameters $m = n + \nu + 1$ and $\nu$. By Lemma \ref{LEMzonal}, we have
\[ \E\left[ \frac{J_\kappa(y ;\beta/2)}{J_\kappa(1^n;\beta/2)} \right] = \frac{J_\kappa(x;\beta/2)}{J_\kappa(1^n;\beta/2)} \E \left[ \frac{J_\kappa(u_1,1^{n-1};\beta/2)}{J_\kappa(1^n;\beta/2)} \right], \quad \quad \kappa\in\Y, \quad \ell(\kappa) \le n \]
where the left hand side expectation is over $y_1,\ldots,y_n$ distributed as $p_n^{(\beta/2)}(y|x)$ and the right hand side expectation is over $u_1$. It turns out that we can extend this identity to arbitrary $\beta > 0$.

\begin{prop} \label{PROPOSITIONJackFactorization}
Fix $\beta > 0$. If $x = (x_1 < \cdots < x_n) \subset [0,1]^n$ and $y = (y_1 < \cdots < y_n)$ is distributed as $p_{n,\nu}(y|x)$, then
\[ \E\left[ \frac{J_\kappa(y ;\beta/2)}{J_\kappa(1^n;\beta/2)} \right] = \frac{J_\kappa(x;\beta/2)}{J_\kappa(1^n;\beta/2)} \E \left[ \frac{J_\kappa(u_1,1^{n-1};\beta/2)}{J_\kappa(1^n;\beta/2)} \right], \quad \quad \kappa\in\Y, \quad \ell(\kappa) \le n \]
where $u_1$ is distributed as $P_{1,\Jacobi}^{(\beta/2)}$ with parameters $m = n + \nu + 1$ and $\nu$.
\end{prop}

\begin{proof}
From the discussion above, $p_{n,\nu}^{(\beta/2)}(y|x)$ is the limit of $K_{\Markov}(\lambda,\mu;A,B)$ under the scaling limit in Proposition \ref{PROPOSITIONMarkovKernel} with $\theta = \beta/2$. Then this proposition is a special case of \cite[Proposition 3.9]{Ahn} (see the discussion following the proof in the reference).
\end{proof}

\section{Crystallization and Gaussianity at $\beta = \infty$} \label{SECbeta=infty}

In this section, we consider the $\beta \to \infty$ limit of the $\beta$-Jacobi product process introduced in Section \ref{SECarbitrarybeta}. We find that the particles freeze at deterministic positions, but fluctuate as correlated Gaussians. We precisely describe this limiting object, the $\infty$-Jacobi product process, in Section \ref{SUBSECgaussian}. This crystallization phenomenon and associated Gaussianity also arise in the closely related $\infty$-Hermite corners process introduced in \cite{GorinMarcus}. We note that the $n\to\infty$ limit of the $\infty$-Hermite corners process was studied in \cite{GorinKleptsyn}. For us, we consider $\beta \to \infty$ and keep $n$ (the number of particles at each step in the Markov chain) fixed. However, we expect that the methods of \cite{GorinKleptsyn} may be generalized to study the $n\to\infty$ limit of the $\infty$-Jacobi product process.

\subsection{Warmup: $\beta \to \infty$ Crystallization of the $\beta$-Jacobi Ensemble} \label{SECwarmup}
Let $\mathcal{U}^n := [0,1]^n \cap \{x_1 \ge \cdots \ge x_n\}$. We can express the density for the $\beta$-Jacobi ensemble \eqref{DISTRIBUTIONJacobi}
\begin{align} \label{density_form}
P_{n,\Jacobi}^{(\beta/2)}(x_1,\ldots,x_n) dx_1 \cdots dx_n = \frac{1}{Z_{n,m,\nu}^{(\beta/2)}} h(x_1,\ldots,x_n;m,\nu)^{\beta/2} \prod_{i=1}^n \frac{dx_i}{x_i(1-x_i)}
\end{align}
supported on $\mathcal{U}^n$ where $Z_{n,m,\nu}^{(\beta/2)}$ is a normalization constant and
\[ h(x_1,\ldots,x_n;m,\nu) = (\triangle(x))^2 \prod_{i=1}^n x_i^{\nu+1} (1 - x_i)^{m - 2n - \nu + 1} \]
for some $m,n,\nu \ge 0$ with $m - 2n - \nu \ge 0$.

By \cite{MarcellanMartinezMartinez}, we know that $(\tilde{x}_1,\ldots,\tilde{x}_n)$ is the unique maximizer of $h(x_1,\ldots,x_n,m,\nu)$ in $\mathcal{U}^n$. It follows that as $\beta \to \infty$, the random particle system $(x_1,\ldots,x_n)$ converges to the deterministic configuration $(\tilde{x}_1,\ldots,\tilde{x}_n)$.

As an example, consider the simple case where $n = 1$, then
\[ h(x_1;m,\nu) = x_1^{\nu+1} (1 - x_1)^{m - \nu - 1} \]
where we need $\nu+1 < m$. The log-derivative gives
\[ \frac{\nu+1}{x_1} - \frac{m - \nu - 1}{1 - x_1} \]
which has a  unique root $\tilde{x}_1 = \tfrac{\nu+1}{m} \in [0,1]$.

\subsection{$\beta \to \infty$ Crystallization of the $\beta$-Jacobi Product Process}
We now turn to the transition kernel \eqref{BetaMarkov}. Fixing $0 < x_1 < \cdots < x_n < 1$, we can see that the transition density has the form
\begin{align} \label{transition_kernel}
p_{n,\nu}^{(\beta/2)}(y|x) = \frac{1}{Z_{x,n,\nu}^{(\beta/2)}} f_n(y;x) g_{n,\nu}(y;x)^{\beta/2}
\end{align}
where $Z_{x,n,\nu}^{(\beta/2)}$ is a normalization constant and
\begin{align*}
f_n(y;x) &= \frac{\triangle(y)}{\prod_{i,j=1}^n |x_i - y_j|} \prod_{i=1}^n \frac{1}{y_i} \\
g_{n,\nu}(y;x) &= \prod_{i,j=1}^n |x_i - y_j| \prod_{i=1}^n y_i^{\nu+1}
\end{align*}
supported in
\[ \mathcal{U}^n_x := [0,1]^n \cap \{ x_1 \le y_1 \le \cdots \le x_n \le y_n\}. \]

As $\beta \to \infty$, we see that $y$ converges in distribution to a deterministic configuration $\tilde{y} = (\tilde{y}_1 \le \ldots \le \tilde{y}_n)$ which satisfies
\[ \tilde{y} = \underset{y}{\mathrm{argmax}}\, g_{n,\nu}(y;x) \]
where the $\mathrm{argmax}$ is over $y \in \mathcal{U}^n_x$. The maximizer solves
\begin{align} \label{g_maximizer}
\sum_{i=1}^n \frac{1}{\tilde{y}_i - x_j} + \frac{\nu+1}{\tilde{y}_i} = 0, \quad \quad i = 1,\ldots,n.
\end{align}
In words, it is a root of the log-gradient of $g_{n,\nu}$; clearly $\log g_{n,\nu}$ is strictly concave on $\mathcal{U}^n_x$ so the maximizer is unique.

\begin{lem} \label{lemMAXIMIZER}
The maximizer $\tilde{y}$ above satisfies
\[ \frac{1}{n} \sum_{i=1}^n (z - \tfrac{\nu +1}{n+\nu+1}x_i) \prod_{j \ne i} (z - x_j) = \prod_{i=1}^n (z - \tilde{y}_i). \]
\end{lem}

\begin{proof}
Following the proof of Theorem 1.1 in \cite{GorinMarcus} by taking $\kappa$ in Proposition \ref{PROPOSITIONJackFactorization} to be a partition of the form $(1,\ldots,1,0,\ldots,0)$, we can show that
\[ \E\left[ \prod_{i=1}^n (z - y_i) \right] = \frac{1}{n!} \sum_{\sigma \in S_n} \E\left[ (z - u_{\sigma(i)} x_i) \right] \]
is independent of $\beta > 0$, where $u_2 = \cdots = u_n = 1$ and $u_1$ is distributed as the $\beta$-Jacobi ensemble with parameters $m = n + \nu + 1$ and $\nu$. We refer the reader to \cite[Section 2]{GorinMarcus} for details on this argument. After evaluating the right hand side expectation, we obtain
\[ \frac{1}{n} \sum_{i=1}^n (z - \E[u_1] x_i) \prod_{j \ne i} (z - x_j). \]
Sending $\beta \to \infty$ and recalling the discussion from \eqref{SECwarmup}, we have
\[ \lim_{\beta \to \infty} \E[u_1] \to \frac{\nu+1}{n+\nu+1} \]
which establishes the lemma. We note that in fact $\E[u_1]$ is independent of $\beta > 0$, see e.g. \cite{Aomoto}.
\end{proof}

\subsection{Gaussianity and the $\infty$-Jacobi Product Process} \label{SUBSECgaussian}

We are now ready to introduce the $\beta \to \infty$ limit of the $\beta$-Jacobi product process.

\begin{defn}
Fix parameters $\nu_k > 0$ for $k = 2,3,\ldots$ and $\tilde{x}^1 = (\tilde{x}^1_1 < \cdots < \tilde{x}^1_n) \in [0,1]^n$. The \emph{$\infty$-Jacobi product process with initial state $\tilde{x}^1$ with parameters $(\nu_k)_{k=2}^\infty$} is a deterministic sequence $(\tilde{x}^k)_{k=1}^\infty$, where $\tilde{x}^k = (\tilde{x}^k_1 < \cdots < \tilde{x}^k_n)$ recursively satisfies
\[ \frac{1}{n} \sum_{i=1}^n (z - \tfrac{\nu_{k+1}+1}{n+\nu_{k+1}+1} \tilde{x}_i^k) \prod_{j \ne i} (z - \tilde{x}_j^k) = \prod_{i=1}^n (z - \tilde{x}_i^{k+1}), \quad \quad k = 2,3,\ldots, \]
equipped with a Gaussian field $\{\xi^k = (\xi_1^k,\ldots,\xi_n^k) \}_{k=1}^\infty$ such that $\xi^1_1 = \cdots \xi^1_n = 0$ and the joint density of $\xi^2,\ldots,\xi^p$ is proportional to
\begin{align} \label{infty_density}
\exp\left[ \sum_{k=2}^p \left( \frac{1}{2} \sum_{1 \le i < j \le n} \frac{(\xi_i^{k-1} - \xi_j^{k-1})^2}{(\tilde{x}_i^{k-1} - \tilde{x}_j^{k-1})^2} - \frac{1}{4} \sum_{i,j=1}^n \frac{(\xi_i^{k-1} - \xi_j^k)^2}{(\tilde{x}_i^{k-1} - \tilde{x}_j^k)^2} + \sum_{i=1}^n \left( \frac{\nu_k+2}{4} \frac{(\xi_i^{k-1})^2}{(\tilde{x}_i^{k-1})^2} - \frac{\nu_k+1}{4} \frac{(\xi_i^k)^2}{(\tilde{x}_i^k)^2} \right) \right) \right].
\end{align}
\end{defn}

\begin{thm}
Fix parameters $\nu_k > 0$ for $k = 2,3,\ldots$ and $x^1 = (x_1^1 < \cdots < x_n^1) \in [0,1]^n$. Suppose $(x^k)_{k=1}^\infty$ is distributed as the $\beta$-Jacobi product process with initial state $x^1$ and parameters $(\nu_k)_{k=2}^\infty$. Let $(\tilde{x}^k)_{k=1}^\infty$ be the deterministic part of the $\infty$-Jacobi product process with initial condition $\tilde{x}^1 = x^1$ and parameters $(\nu_k)_{k=2}^\infty$. If $x^k = \tilde{x}^k + \frac{\Delta x^k}{\sqrt{\beta}}$, then $(\tilde{x}^k)_{k=2}^\infty$, $(\Delta x^k)_{k=1}^\infty$ converges to the Gaussian field $(\xi^k)_{k=1}^\infty$ associated to the $\infty$-Jacobi product process with initial state $\tilde{x}^1$ as $\beta \to \infty$.
\end{thm}

\begin{proof}
We proceed as in \cite[Section 3.4]{GorinMarcus}. The density of $(x^1,\ldots,x^p)$ is proportional to
\[ \prod_{k=2}^p \frac{\triangle(x^k)}{(\triangle(x^{k-1}))^{\beta - 1}} \prod_{i,j=1}^n |x_i^{k-1} - x_j^k|^{\beta/2 - 1} \prod_{i=1}^n \frac{(x_i^k)^{\beta/2(\nu_k + 1) - 1}}{(x_i^{k-1})^{\beta/2(\nu_k + 2) - 1}} dx^k_i. \]
Sequentially define $\tilde{x}^1 = x^1$ and $\tilde{x}^k$ to be the maximizer of $g_{n,\nu_k}(x^k;\tilde{x}^{k-1})$ for $k = 2,3,\ldots$. Observe that
\begin{align} \label{recursion_zeros}
\prod_{i=1}^n (z - \tilde{x}_i^{k+1}) = \frac{1}{n} \sum_{i=1}^n (z - \tfrac{\nu_{k+1}+1}{n+\nu_{k+1}+1} \tilde{x}_i^k) \prod_{j\ne i} (z - \tilde{x}_j^k), \quad \quad k = 2,3,\ldots
\end{align}
by Lemma \ref{lemMAXIMIZER}. In other words $(\tilde{x}^k)_{k=1}^\infty$ is the deterministic part of a $\infty$-Jacobi product process. Set
\[ x^k_i = \tilde{x}^k_i + \frac{\Delta x^k_i}{\sqrt{\beta}}. \]
The density (ignoring the differentials) then becomes
\begin{align*}
& \frac{1}{Z_n(\beta)} \prod_{k=2}^p \left( \prod_{i < j} \frac{|\tilde{x}_i^k - \tilde{x}_j^k|}{|\tilde{x}_i^{k-1} - \tilde{x}_j^{k-1}|^{\beta-1}} \prod_{i,j=1}^n |\tilde{x}_i^{k-1} - \tilde{x}_j^k|^{\beta/2 - 1} \prod_{i=1}^n \frac{(\tilde{x}_i^k)^{\beta/2(\nu_k+1)-1}}{(\tilde{x}_i^{k-1})^{\beta/2(\nu_k+2)-1}} \right) \\
& \quad \times \prod_{k=2}^p \left( \prod_{i < j} \frac{\left| 1 + \frac{\Delta x_i^k - \Delta x_j^k}{\sqrt{\beta}(\tilde{x}_i^k - \tilde{x}_j^k)} \right|}{\left| 1 + \frac{\Delta x_i^{k-1} - \Delta x_j^{k-1}}{\sqrt{\beta}(\tilde{x}_i^{k-1} - \tilde{x}_j^{k-1})} \right|^{\beta - 1}} \prod_{i,j=1}^n \left|1 + \frac{\Delta x_i^{k-1} - \Delta x_j^k}{\sqrt{\beta}(\tilde{x}_i^{k-1} - \tilde{x}_j^k)} \right|^{\beta/2-1} \prod_{i=1}^n \frac{\left| 1 + \frac{\Delta x_i^k}{\sqrt{\beta} \tilde{x}_i^k} \right|^{\beta/2(\nu_k+1) - 1}}{\left| 1 + \frac{\Delta x_i^{k-1}}{\sqrt{\beta} \tilde{x}_i^{k-1}} \right|^{\beta/2(\nu_k+2)-1}} \right)
\end{align*}
The second line can be rewritten as
\begingroup\makeatletter\def\f@size{10}\check@mathfonts
\begin{align*}
& \prod_{k=2}^p \exp\left[ - \sqrt{\beta} \sum_{i < j} \frac{\Delta x_i^{k-1} - \Delta x_j^{k-1}}{\tilde{x}_i^{k-1} - \tilde{x}_j^{k-1}} + \frac{\sqrt{\beta}}{2} \sum_{i,j=1}^n \frac{\Delta x_i^{k-1} - \Delta x_j^k}{\tilde{x}_i^{k-1} - \tilde{x}_j^k} + \frac{\sqrt{\beta}}{2} \sum_{i=1}^n \left( (\nu_k+1) \frac{\Delta x_i^k}{\tilde{x}_i^k} - (\nu_k+2) \frac{\Delta x_i^{k-1}}{\tilde{x}_i^{k-1}} \right) \right] \\
& \quad \times \prod_{k=2}^p \exp\left[ \frac{1}{2} \sum_{i < j} \frac{(\Delta x_i^{k-1} - \Delta x_j^{k-1})^2}{(\tilde{x}_i^{k-1} - \tilde{x}_j^{k-1})^2} - \frac{1}{4} \sum_{i,j=1}^n \frac{(\Delta x_i^{k-1} - \Delta x_j^k)^2}{(\tilde{x}_i^{k-1} - \tilde{x}_j^k)^2} - \frac{\nu_k+1}{4} \frac{(\Delta x_i^k)^2}{(\tilde{x}_i^k)^2} + \frac{\nu_k+2}{4} \frac{(\Delta x_i^{k-1})^2}{(\tilde{x}_i^{k-1})^2} + O(1/\sqrt{\beta}) \right].
\end{align*} \endgroup
The second line of the above gives the density of the $\infty$-Jacobi product process up to normalization as $\beta \to \infty$. We check that the first line evaluates to $1$. The coefficient of $\Delta x_i^k$ is given by
\begin{align} \label{cancel}
- \sqrt{\beta} \sum_{j \ne i} \frac{1}{\tilde{x}_i^k - \tilde{x}_j^k} + \frac{\sqrt{\beta}}{2} \sum_{j=1}^n \frac{1}{\tilde{x}_i^k - \tilde{x}_j^{k+1}} - \frac{\sqrt{\beta}}{2} \frac{\nu_{k+1}+2}{\tilde{x}_i^k} + \frac{\sqrt{\beta}}{2} \sum_{j=1}^n \frac{1}{\tilde{x}_i^k - \tilde{x}_j^{k-1}} + \frac{\sqrt{\beta}}{2} \frac{\nu_k+1}{\tilde{x}_i^k}
\end{align}
for $k = 2,\ldots,p-1$, and for $k = p$ the first three terms are omitted. The last two terms cancel by the critical point equation \eqref{g_maximizer} for $1 \le k \le p$. To see that the first three terms cancel for $1 \le k < p$, differentiate \eqref{g_maximizer} and evaluate at $z = \tilde{x}_i^k$ to obtain
\begin{align} \label{chareq_differentiate}
\begin{split}
\sum_{j=1}^n \prod_{a \ne j} (\tilde{x}_i^k - \tilde{x}_a^{k+1}) &= \frac{1}{n} \left( \prod_{j \ne i} (\tilde{x}_i^k - \tilde{x}_j^k) + \sum_{a \ne i} \left( \tfrac{2n+\nu_{k+1}+1}{n+\nu_{k+1}+1} \tilde{x}_i^k - \tfrac{\nu_{k+1}+1}{n+\nu_{k+1}+1} \tilde{x}_a^k \right) \prod_{j\ne a,i}(\tilde{x}_i^k - \tilde{x}_j^k) \right) \\
&= \tfrac{\nu_{k+1}+2}{n+\nu_{k+1}+1} \prod_{j \ne i} (\tilde{x}_i^k - \tilde{x}_j^k) + \tfrac{2}{n+\nu_{k+1}+1} \tilde{x}_i^k \sum_{a \ne i} \prod_{j\ne a,i}(\tilde{x}_i^k - \tilde{x}_j^k)
\end{split}
\end{align}
Evaluating \eqref{g_maximizer} at $z = \tilde{x}_i^k$, we obtain
\begin{align} \label{chareq_evaluate}
\prod_{j=1}^n (\tilde{x}_i^k - \tilde{x}_j^{k+1}) = \frac{\tilde{x}_i^k}{n+\nu_{k+1}+1} \prod_{j \ne i} (\tilde{x}_i^k - \tilde{x}_j^k)
\end{align}
Divide \eqref{chareq_differentiate} by \eqref{chareq_evaluate} to obtain
\[ \sum_{j=1}^n \frac{1}{\tilde{x}_i^k - \tilde{x}_j^{k+1}} = (\nu_{k+1}+2)\frac{1}{\tilde{x}_i^k} + 2 \sum_{j \ne i} \frac{1}{\tilde{x}_i^k - \tilde{x}_j^k} \]
from which we see that the first three terms of \eqref{cancel} cancel.

It remains to check the integrability of \eqref{infty_density}. This would be immediate if the coefficients of each of the $\xi_i^k$ were negative, however this is not the case. From the integrability of $p_{n,\nu}^{(\beta/2)}(y|x)$, we have
\[ \frac{1}{(\Gamma(\theta))^n} \frac{\Gamma(\theta(\nu+n+1))}{\Gamma(\nu+1)} \underset{y \in \mathcal{U}_x^n}{\int \!\!\cdots \!\!\int} \triangle(y) \prod_{i,j=1}^n |x_i - y_j|^{\beta/2-1} \prod_{i=1}^n y_i^{\beta/2(\nu+1)-1} \,dy_i = \triangle(x)^{\beta-1} \prod_{i=1}^n x_i^{\beta/2(\nu+1)-1}. \]
Arguing as above, we see that
\begin{align*}
& \int \! \! \cdots \! \! \int \exp\left( -\frac{1}{4} \sum_{j=1}^n \frac{(\zeta_i - \xi_j)^2}{(x_i^{k-1} - x_j^k)^2} - \frac{\nu_k+1}{4} \frac{(\xi_i)^2}{(x_i^k)^2} \right) d \xi \\
& \quad \quad = Z \cdot \exp\left(- \frac{1}{2} \sum_{i < j} \frac{(\zeta_i - \zeta_j)^2}{(x_i^{k-1} - x_j^{k-1})^2} - \frac{\nu+1}{4} \frac{\zeta_i^2}{(x_i^{k-1})^2} \right)
\end{align*}
for some constant $Z > 0$ independent of $\zeta_1,\ldots,\zeta_n$. Thus we can integrate the density of the $(\xi^2,\ldots,\xi^p)$ in the $\infty$-Jacobi corners process sequentially in $k$ starting from $p$ and descending using the identity above.
\end{proof}


\end{document}